\documentclass[journal]{IEEEtran}

\usepackage{amsmath}
\usepackage{graphicx}
\usepackage{threeparttable}
\usepackage{amssymb}
\usepackage{amsthm}
\usepackage{makecell}
\usepackage{subfigure}
\usepackage{cite}
\usepackage{multirow}

\newtheorem{lemma}{Lemma}

\newtheorem{proposition}{Proposition}

\newtheorem{remark}{Remark}
\usepackage{color}

\usepackage{hyperref}
\hypersetup{
	colorlinks=true,
	linkcolor=blue,
	filecolor=red,      
	urlcolor=blue,
	citecolor=blue,
}

\ifCLASSINFOpdf

\else

\fi

\begin{document}

\title{Modular PE-Structured Learning for Cross-Task Wireless Communications}


\author{Yuxuan Duan and Chenyang Yang}

\vspace{-5mm}
\maketitle
\begin{abstract}
Recent trends in learning wireless policies attempt to develop deep neural networks (DNNs) for handling multiple tasks with a single model. Existing approaches often rely on large models, which are hard to pre-train and fine-tune at the wireless edge.  In this work, we challenge this paradigm by leveraging the structured knowledge of wireless problems -- specifically, permutation equivariant (PE) properties. We design three types of PE-aware modules, two of which are Transformer-style sub-layers. These modules can serve as building blocks to assemble compact DNNs applicable to the wireless policies with various PE properties. To guide the design, we analyze the hypothesis space associated with each PE property, and show that the PE-structured module assembly can boost the learning efficiency.
Inspired by the reusability of the modules, we propose PE-MoFormer, a compositional DNN capable of learning a wide range of wireless policies  -- including but not limited to precoding, coordinated beamforming, power allocation, and channel estimation -- with strong generalizability, low sample and space complexity. Simulations demonstrate that the proposed modular PE-based framework outperforms relevant large model in both learning efficiency and inference time, offering a new direction for structured cross-task learning for wireless communications.
\end{abstract}
\vspace{-1mm}

\IEEEpeerreviewmaketitle

\vspace{-2mm}\section{Introduction}\label{sec:intro}\vspace{-1mm}
Deep neural networks (DNNs) have been developed for a multitude of wireless problems, such as resource allocation and channel acquisition. Conventional deep learning  approaches train a DNN for each specific problem. Despite its good performance, the large number of wireless problems makes it impractical to train and store the specific DNNs.

In the past several years, large language models (LLMs), with a backbone of Transformer, have attracted widespread attention \cite{LLMSurvey}. These LLMs are pre-trained and then fine-tuned for various downstream tasks. Inspired by their versatility, recent research efforts have been made to leverage pre-trained large models for multiple wireless tasks \cite{DaiLL_LLMMultitask,LLM4WMchenxiang,wirelessGPT}. However, applying large models to wireless communications faces several challenges, say the high sample, memory, and computing demands, the limited adaptability to dynamic wireless systems, and the lack of structured communication knowledge \cite{LargeAI,ReflectionOnGPT}.

To tackle the difficulties, a promising direction is resorting to generic wireless knowledge. For example, permutation equivariance (PE) properties, though still underappreciated, exist in a variety of wireless problems. By leveraging PE properties for constraining hypothesis space (i.e., the functions representable by a DNN), DNNs often exhibit lower complexity and better generalizability to dynamic problem sizes (e.g., the number of users)\cite{zhao2022understanding, shenyifei, Gformer}. Nonetheless, how to design a cross-task DNN (say Transformer) by incorporating PE properties is still open.

\vspace{-2mm}\subsection{Related Works}\vspace{-0.5mm}
\subsubsection{Transformers for Wireless Communications}\vspace{-0.1mm}
Transformer was originally proposed for sequential tasks \cite{vaswani2017attention}, whose input and output are sequences of tokens. The core of the Transformer is the attention mechanism, which can capture the dependence among tokens.

Transformer has been introduced to optimize wireless problems \cite{MagazineTransformer}, mostly for channel acquisition tasks including channel prediction \cite{CPdai, LinFormer},  estimation \cite{trans_CE_payless_TWC}, and feedback \cite{transCSIfeed4}. For example, in \cite{CPdai}, the channels in past time slots were treated as tokens, which are input into a Transformer for channel prediction, using the attention mechanism to capture the temporal correlation.

Transformer has also been integrated with other DNNs for resource allocation \cite{Kansformer, mehrabian2024joint, TransformerIndirect, AdaTTD_NearField, li2024hpe, Gformer}. For example, graph Transformers were designed in \cite{Gformer} to learn baseband and hybrid precoding in multi-user multi-input-single-output (MU-MISO) systems for maximizing spectrum efficiency (SE). Each user was treated as a token, enabling the attention to model multi-user interference (MUI). In \cite{Kansformer}, a Transformer combined with a Kolmogorov-Arnold network was designed to optimize precoding for maximizing the energy efficiency of MU-MISO systems.  In \cite{mehrabian2024joint}, a graph Transformer was designed to jointly optimize bandwidth allocation, precoding at the base station (BS), and coefficients of intelligent reflecting surfaces (IRS).

All these DNNs are task-specific, i.e., each DNN is designed and trained for a single problem.

\subsubsection{DNNs with Permutation Properties}
A diverse array of wireless optimization problems consist of sets \cite{liu2023multidimensional}, say signal detection \cite{Detection}, power/bandwidth allocation \cite{guo2022het_pc, Satellite_1D-PEGNN2024, ENGNNLiyang_2024,Sunchengjian}, and precoding \cite{zhao2022understanding, liu2023multidimensional,mehrabian2024joint,HybridBeamforming_1D-PEGNN2024,GAT}.
As a result, the corresponding \emph{wireless policies}, each of which is a mapping from environment parameters to optimized variables, satisfy permutation properties\cite{Sunchengjian}.
There are many types of permutation properties. For example, the MU-MISO precoding policy satisfies an independent two-dimensional (2D)-permutation equivariance (PE) property \cite{zhao2022understanding}, and the power control for device-to-device (D2D) communications satisfies a joint 2D-PE property \cite{shenyifei}.
The permutation properties of wireless policies can be identified by analyzing the types of sets and their relations in wireless problems with the method in \cite{liu2023multidimensional}.

By leveraging the PE properties as prior, the training complexity of the designed DNN can be reduced remarkably. Moreover, a DNN satisfying the PE property induced by a set has the potential for generalizing to the size of the set \cite{guo2024recursive}. With this size-generalizability, the well-trained DNN can be used for inference in problems of different scales without re-training, which is desirable for dynamic wireless systems.

For example, both the GNNs designed in \cite{HybridBeamforming_1D-PEGNN2024} and in \cite{GAT} satisfy the PE property induced by the user set,
which can be generalized to the number of users.
Yet these works only employ the PE property induced by a part of the sets of the considered problems, while the permutability induced by other sets, say the set of antennas (ANs) at the BS, is neglected. In \cite{liu2023multidimensional}, a unified framework was proposed for designing a graph convolutional network (GCN) whose input-output relation has exactly the same PE property of a wireless policy (i.e., \emph{the GCN has matched PE property to the policy}). By exploiting the permutations induced by all sets of a problem, the designed GCNs therein were validated to have lower sample complexity and be generalizable to the sizes of all these sets.

Another example is the tailored Transformers, which have been designed to satisfy simple PE properties. In \cite{SetTransformer2019}, a Set-Transformer was designed to satisfy a one-dimensional (1D)-PE property by removing order-sensitive components in Transformers like positional encoding. In \cite{li2024hpe}, a hierarchical Transformer (Hier-Former) was proposed to satisfy a nested 1D-PE property of a multi-cast policy in MU-MISO systems.
In \cite{Gformer}, a 2D-Former and three-dimensional (3D)-Former were designed to satisfy the independent 2D- and 3D-PE properties.
 However, these Transformer-based DNNs are unable to satisfy more complex PE properties, say the partial nested 2D-PE property of multi-user multi-input-multi-output
(MU-MIMO) precoding policy. With mismatched property, a DNN has high training complexity or degraded learning performance.

Again, each of these DNNs was designed to satisfy a single type of PE property, restricting their efficiency for multiple policies with different PE properties.

\subsubsection{DNNs for Multiple Tasks}
Since training a specific DNN for every task is costly, several studies strived to design DNNs for more than one wireless task \cite{Yuwei_Transferlearning_2024,AiBo_Tcom_PlugPlay}. For example, in \cite{Yuwei_Transferlearning_2024}, a DNN was proposed for a couple of wireless tasks with the same input. Specifically, a shared fully-connected neural network (FNN) was designed to extract common representations across tasks, which is cascaded with private FNNs to produce task-specific outputs. This DNN can be applied to problems with different optimization objectives, such as maximizing the sum rate or the minimum rate of users. Nonetheless, as a kind of multi-output learning \cite{SurveyonMultitaskLearning}, the DNN can only be applied to tasks with the same input.

Encouraged by the success of LLMs in handling diverse tasks, recent works resorted to large models for multiple wireless problems. In \cite{DaiLL_LLMMultitask}, a pre-trained LLM (GPT-2 or LLAMA2) was fine-tuned to share across three tasks: channel prediction, signal detection, and MU-MISO precoding. For each task, task-specific encoder and decoder were designed before and after the LLM. In \cite{LLM4WMchenxiang}, a pre-trained GPT-2 combined with task-specific FNNs or convolutional neural networks was designed for five tasks: channel estimation and prediction, beam selection, and estimation of distance and pathloss.
In \cite{wirelessGPT}, a Transformer with over 80M parameters was pre-trained using samples of wireless channels for extracting common channel representations, which is then cascaded with task-specific DNNs to produce outputs for downstream tasks.

Despite the multi-task capability, the upstream or downstream DNNs ought to be designed and trained for every task. Moreover, a large number of training samples (ranging from 14.8k to 50k for each task in \cite{DaiLL_LLMMultitask,LLM4WMchenxiang}) and model parameters (ranging from 80M to 7B in \cite{DaiLL_LLMMultitask,LLM4WMchenxiang,wirelessGPT}) are required, which makes pre-training and online fine-tuning challenging.

\vspace{-4mm}\subsection{Motivation and Contributions}\vspace{-1mm}
Existing DNNs for multiple wireless tasks are limited to the tasks with a single input, or suffer from high complexity for pre-training, fine-tuning, and inference. Moreover, the size-generalizability of these DNNs is rarely considered.

In this paper, we challenge the prevailing large model paradigm with structured communication knowledge. In particular, we conceive PE-MoFormer, a DNN with \textbf{PE}-aware \textbf{mo}dules including Trans\textbf{former}-style sub-layers, to enable efficient learning across a range of wireless problems. The difficulties in properly designing and assembling modules are addressed by comprehending PE-structured DNNs and wireless tasks.
Our major contributions are as follows.
\begin{itemize}
    \item \textbf{PE-guided modular design and assembly:} We design three modules, each satisfying a specific PE property. By proving that the hypothesis space of a DNN depends on the layer with the largest hypothesis space, we show that these modules and the modules in \cite{Gformer} can be assembled into multiple specific PE-Formers.

    \item \textbf{Cross-task compositional DNN:} We propose a PE-MoFormer, which enables cross-task learning on diverse wireless tasks by composing appropriate modules to satisfy various PE properties. This design is inspired by the observation that different PE properties can be satisfied by reusing a small set of modules, and the same structural property is shared among more than one wireless policy.

    \item \textbf{High learning efficiency:} For a proof-of-concept demonstration, we use simulations to evaluate the proposed PE-MoFormer on the tasks of channel estimation, precoding, coordinated beamforming (CB), and power allocation.
        The results show that the PE-MoFormer can achieve good performance across the tasks with only 100 training samples per task and 9.71k parameters in total. Moreover, it has strong generalizability to problem sizes, signal-to-noise ratios (SNRs), and channel statistics, exemplifying its adaptability to dynamic wireless environments.

\end{itemize}

The rest of this paper is organized as follows: In section \ref{sec:PE properties}, we introduce several PE properties and wireless policies with these properties. In section \ref{sec:1D-2D-Former}, we recap two existing Transformer-based DNNs with PE properties. In section \ref{sec:methodology}, we analyze the hypothesis space of DNNs, design three modules and assemble them into specific PE-Formers, and propose the cross-task  PE-MoFormer. Simulations and conclusions are provided in sections \ref{sec:simulations} and \ref{sec:conclusion}.

\textit{Notations}: $\mathbf{X} = [x_{i,j}]$ or $\mathbf{X} = [\mathbf{X}_{i,j}]$ denotes a matrix with the element or sub-matrix in the $i$-th row $j$-th column being $x_{i,j}$ or $\mathbf{X}_{i,j}$. $\mathbf{x}_i$ represents the $i$-th column vector of matrix $\mathbf{X}$. $\mathbf{X} = \mathrm{diag}(\mathbf{x}_1,\cdots,\mathbf{x}_N)$ denotes a diagonal matrix with the $i$-th diagonal block being $\mathbf{x}_i$. $\mathbf{I}_{N}$ denotes an $N\times N $ identity matrix where the subscript $N$ may be omitted for brevity. $\mathsf{Tr}(\cdot)$ denotes the trace of a matrix. $(\cdot)^{\mathsf{T}}$ and $(\cdot)^{\mathsf{H}}$ denote the transpose and conjugate transpose operations, respectively. $\otimes$ and $\odot$ represent the Kronecker product and Hadamard product of a matrix, respectively. $\mathbb{R}$ and $\mathbb{C}$ denote the real and complex number fields, and $\mathsf{Re}(\cdot)$ and $\mathsf{Im}(\cdot)$ represent the real and imaginary parts of a complex vector or matrix, respectively. $\subsetneq$ stands for proper inclusion between sets.

\vspace{-1mm}\section{PE Properties and Wireless Policies}\label{sec:PE properties}\vspace{-0.1mm}
In this section, we introduce several PE properties that have been found in wireless communications and present several wireless policies to demonstrate the PE structured knowledge.

\vspace{-2mm}\subsection{PE Properties of Functions Defined on Different Sets}\label{sec:PEs}\vspace{-0.5mm}
The functions defined on sets have PE properties \cite{zaheer2017deep}.
\subsubsection{Sets}
A \textbf{normal set} is composed of unordered elements. A \textbf{nested set} is a set of subsets, where the elements in every subset and the subsets can be permuted arbitrarily. For example, $\{a_1, a_2,a_3\}$ is a normal set where the elements $a_i$ and $a_j$ can be arbitrarily permuted. $\{\{a_1,a_2,a_3\},\{b_1,b_2\}\}$ is a nested set with two subsets, where $a_i$ and $a_j$ (or $b_1$ and $b_2$) can be arbitrarily permuted, and the two subsets can be swapped. However, permutation among the elements in different subsets (e.g., permuting $a_i$ and $b_j$) is not allowed.

\subsubsection{PE Properties}
A function may be defined over more than one set. Depending on the types of sets and the relation between the sets, the functions on sets exhibit different PE properties. A large number of wireless policies are defined on sets \cite{liu2023multidimensional}, whose inputs and outputs are feature vectors/matrices/tensors of the sets \cite{Sunchengjian}. For easy understanding with simpler notations, we only present eight PE properties.

A function defined on \emph{a single set} satisfies one of two PE properties as follows, depending on the type of the set.

\begin{subequations}
\textbf{1D-PE property:}
\vspace{-2mm}
\begin{align}\label{eq:pre_1d_pe}
 \mathbf{Y}\mathbf{\Pi} = f(\mathbf{X}\mathbf{\Pi})
\end{align}

\vspace{-2mm}
\textbf{Nested 1D-PE property:}
\vspace{-2mm}
\begin{align}\label{eq:pre_nested_1d_pe}
    &\quad\quad\quad\quad\quad\mathbf{Y}\mathbf{\Omega} = f(\mathbf{X}\mathbf{\Omega}),\\\notag&\mathbf{\Omega} = (\mathbf{\Pi}_{\mathsf{sub}}\otimes\mathbf{I}_{N_\mathsf{sub}})\mathrm{diag}(\mathbf{\Pi}_{1},\cdots,\mathbf{\Pi}_{N_\mathsf{sub}})
\end{align}
\end{subequations}
 where $\mathbf{\Pi}$ is an arbitrary permutation matrix,\footnote{By arbitrary, we mean  $\mathbf{\Pi}$ has $N!$ possible values for a set with $N$ elements. For instance, a function defined on a normal set with two elements, $\mathbf{Y}\triangleq [\mathbf{y}_1, \mathbf{y}_2] = f([\mathbf{x}_1,\mathbf{x}_2]) \triangleq f(\mathbf{X})$, has two possible permutation matrices, where $\mathbf{Y}$ and $\mathbf{X}$ are the feature matrices of the set. One possible permutation matrix is $\mathbf{\Pi} = \begin{bmatrix}0,1\\1,0\end{bmatrix}$, which permutes the order of the elements in $\mathbf{X}$ as $[\mathbf{x}_2,\mathbf{x}_1]$ and the order of the elements in $\mathbf{Y}$ as $[\mathbf{y}_2,\mathbf{y}_1]$. The other is $\mathbf{\Pi} = \mathbf{I}_{2}$.} $\mathbf{\Omega}$ is a nested permutation matrix comprises two levels of permutation (i.e., permuting the elements in each subset with  $\mathbf{\Pi}_{i},i=1,\cdots,N_{\mathsf{sub}}$ and permuting the subsets with $\mathbf{\Pi}_{\mathsf{sub}}$), and
 $N_{\mathsf{sub}}$ is the number of subsets in the nested set.

 When a function is defined on \emph{a normal set},
 it satisfies the \emph{1D-PE property} in \eqref{eq:pre_1d_pe}.
When a function is defined on \emph{a nested set},
it satisfies the \emph{nested 1D-PE property} in \eqref{eq:pre_nested_1d_pe}.

A function defined on \emph{two sets} satisfies one of six 2D-PE properties as follows, depending on the types of the sets and the relation between the sets.
\begin{subequations}

\textbf{Ind. 2D-PE property:}
\vspace{-1mm}\begin{align}\label{eq:pre_2d_pe}
\mathbf{\mathbf{\Pi}}_\mathsf{A}^\mathsf{T}\mathbf{Y}\mathbf{\mathbf{\Pi}}_\mathsf{B} = f(\mathbf{\mathbf{\Pi}}_\mathsf{A}^\mathsf{T}\mathbf{X}\mathbf{\mathbf{\Pi}}_\mathsf{B})
\end{align}

\vspace{-1mm}
\textbf{Joint 2D-PE property:}
\vspace{-1mm}
\begin{align}\label{eq:pre_joint_2d_pe}
\mathbf{\mathbf{\Pi}}^\mathsf{T}_{\mathsf{A}}\mathbf{Y}\mathbf{\mathbf{\Pi}}_{\mathsf{B}} = f(\mathbf{\mathbf{\Pi}}^\mathsf{T}_\mathsf{A}\mathbf{X}\mathbf{\mathbf{\Pi}}_{\mathsf{B}}),  \text{  if }\mathbf{\mathbf{\Pi}}_{\mathsf{A}} = \mathbf{\mathbf{\Pi}}_{\mathsf{B}}
\end{align}

\vspace{-2mm}
\textbf{Partial-nested 2D-PE property:}
\vspace{-2mm}
\begin{align}\label{eq:pre_partial_nested_2d_pe}
\mathbf{\mathbf{\Pi}}_{\mathsf{A}}^\mathsf{T}\mathbf{Y}\mathbf{\mathbf{\Omega}} = f(\mathbf{\mathbf{\Pi}}_{\mathsf{A}}^\mathsf{T}\mathbf{X}\mathbf{\mathbf{\Omega}})
\end{align}

\textbf{Nested 2D-PE properties:}

\vspace{1mm}
\hspace{5mm}\textbf{Nested ind. 2D-PE property}
\vspace{-1.5mm}
\begin{align}\label{eq:pre_nested_indepedent_2D_pe}\mathbf{\mathbf{\Omega}}_\mathsf{A}^\mathsf{T}\mathbf{Y}\mathbf{\mathbf{\Omega}}_\mathsf{B} = f(\mathbf{\mathbf{\Omega}}_\mathsf{A}^\mathsf{T}\mathbf{X}\mathbf{\mathbf{\Omega}}_\mathsf{B})
\end{align}

\vspace{-1.5mm}
\hspace{5mm}\textbf{Nested partial-joint 2D-PE property}
\vspace{-1.5mm}
\begin{align}\label{eq:pre_nested_partial_joint_2d_pe}\mathbf{\mathbf{\Omega}}_\mathsf{A}^\mathsf{T}\mathbf{Y}\mathbf{\mathbf{\Omega}}_\mathsf{B} = f(\mathbf{\mathbf{\Omega}}_\mathsf{A}^\mathsf{T}\mathbf{X}\mathbf{\mathbf{\Omega}}_\mathsf{B}),\text{  if }\mathbf{\mathbf{\Pi}}_{\mathsf{A,sub}} = \mathbf{\mathbf{\Pi}}_{\mathsf{B,sub}}
\end{align}

\vspace{-1.5mm}
\hspace{5mm}\textbf{Nested joint 2D-PE property}
\vspace{-1.5mm}
\begin{align}\label{eq:pre_nested_joint_2d_pe}\mathbf{\mathbf{\Omega}}_\mathsf{A}^\mathsf{T}\mathbf{Y}\mathbf{\mathbf{\Omega}}_\mathsf{B} = f(\mathbf{\mathbf{\Omega}}_\mathsf{A}^\mathsf{T}\mathbf{X}\mathbf{\mathbf{\Omega}}_\mathsf{B}),\text{  if }\mathbf{\mathbf{\Omega}}_{\mathsf{A}} = \mathbf{\mathbf{\Omega}}_{\mathsf{B}}
\end{align}

\end{subequations}\vspace{-2mm}
where the two nested permutation matrices are
\vspace{-3mm}
\begin{align}
& \mathbf{\Omega}_\mathsf{A} =(\mathbf{\Pi}_{\mathsf{A,sub}}\otimes\mathbf{I}_{N_\mathsf{sub}})\mathrm{diag}(\mathbf{\Pi}_{\mathsf{A},1},\cdots,\mathbf{\Pi}_{\mathsf{A},N_\mathsf{sub}})\notag
\\
& \mathbf{\Omega}_\mathsf{B} =(\mathbf{\Pi}_{\mathsf{B,sub}}\otimes\mathbf{I}_{N_\mathsf{sub}})\mathrm{diag}(\mathbf{\Pi}_{\mathsf{B},1},\cdots,\mathbf{\Pi}_{\mathsf{B},N_\mathsf{sub}})\notag
\end{align}

When a function is defined on \emph{two normal sets}, it satisfies the property in \eqref{eq:pre_2d_pe} or \eqref{eq:pre_joint_2d_pe}.
If the permutations of the two sets are independent (i.e., the elements in each set can be permuted disregarding the permutation of the elements in another set), the function satisfies the \emph{independent (ind.) 2D-PE property}.
If the elements in a set should be permuted together with the elements in the other set (i.e., permuting jointly), the function satisfies the \emph{joint 2D-PE property}.

When a function is defined on \emph{a normal set and a nested set}, it satisfies a \emph{partial-nested 2D-PE property} in \eqref{eq:pre_partial_nested_2d_pe}.

When a function is defined on two nested sets,  it satisfies the property in \eqref{eq:pre_nested_indepedent_2D_pe}, \eqref{eq:pre_nested_partial_joint_2d_pe}, or \eqref{eq:pre_nested_joint_2d_pe}. If the permutation of a nested set is independent of the other, the function satisfies the \emph{nested ind. 2D-PE property}. If the subsets in the two nested sets should be permuted jointly, but the elements in each subset of a nested set can be permuted independently from the permutation of the elements in another nested set, the function has the \emph{nested partial-joint 2D-PE property}.
If both the subsets and the elements in each subset should be permuted jointly, the function has the \emph{nested joint 2D-PE property}.

Furthermore, a function defined on more than two sets satisfies different high-dimensional PE properties, again depending on the types of the sets and the relation
among the sets.

For example, a function defined on three normal sets has the \emph{ind. 3D-PE property} as follows, if the permutation in each set is independent of the permutations of the other sets.

\textbf{Ind. 3D-PE property:}
\vspace{-2mm}
\begin{equation}
\begin{aligned}\label{eq:pre_3d_pe}
\boldsymbol{Y}_{\pi_{\mathsf{A}}(i),\pi_\mathsf{B}(j),\pi_{\mathsf{C}}(k)} = f(\boldsymbol{X}_{\pi_{\mathsf{A}}(i),\pi_\mathsf{B}(j),\pi_{\mathsf{C}}(k)})
\end{aligned}
\end{equation}
where $\boldsymbol{X}$ and $\boldsymbol{Y}$ are tensors with three dimensions, the slice along each dimension is the feature of each set, and $\pi(\cdot)$ denotes the permutation operation on a set.

\vspace{-1mm}\begin{remark}
Some wireless policies have permutation invariance (PI) properties, which are not analyzed for conciseness.
\end{remark}

\vspace{-4mm} \subsection{Wireless Policies with Various PE Properties}\label{w-policies}\vspace{-0.5mm}

To avoid presenting multiple optimization problems, we provide an optimization problem for a relatively general system setting, which can degenerate into different problems.
Specifically, consider a system with $M$ cells that share the same spectrum. In the $m$-th cell (denoted as $\text{Cell}_m$), a BS (denoted as $\text{BS}_m$) serves multiple user equipments (UEs). Each BS is equipped with $N_{\mathsf{t}}$ antennas (ANs) and serves $K$ UEs, where each UE is equipped with $N_{\mathsf{r}}$ ANs.

The precoding matrices at all BSs can be jointly optimized, say to maximize SE \cite{wmmse}, from the following problem\vspace{-2mm}
\begin{equation}\label{eq:multi-cell problem mul-users}
\begin{aligned}
&\mathrm{P1}:\max_{\{\mathbf{V}_{m_k}\}}\sum_{m=1}^{M}\sum_{k=1}^{K}\log_{2}\mathrm{det}\Bigg(\mathbf{I}+\mathbf{H}_{m,m_k}^{\mathsf{H}}\mathbf{V}_{m_k}\mathbf{V}_{m_k}^{\mathsf{H}}\mathbf{H}_{m,m_k}\\&\quad\quad\quad\quad\bigg(\sum_{(i,j)\neq(m,k)}\mathbf{H}_{i,m_k}^{\mathsf{H}}\mathbf{V}_{i_j}\mathbf{V}_{i_j}^{\mathsf{H}}\mathbf{H}_{i,m_k}+\sigma_{0}^{2}\mathbf{I}\bigg)^{-1}\Bigg)\\
    &\quad\quad\,\,\mathrm{s.t.}\quad\quad\quad\sum_{k=1}^{K}\mathsf{Tr}(\mathbf{V}_{m_k}^{\mathsf{H}}\mathbf{V}_{m_k}) \leq P_t, \quad\quad\forall m
 \end{aligned}
\end{equation}
where $\mathbf{V}_{m_k} \in \mathbb{C}^{N_\mathsf{t} \times N_{\mathsf{r}}}$ and  $\mathbf{H}_{i,m_k} \in \mathbb{C}^{N_\mathsf{t} \times N_{\mathsf{r}}}$ are respectively the precoding matrix for the $k$-th UE in $\text{Cell}_m$ (denoted as $\text{UE}_{m_k}$) and the channel matrix from $\text{BS}_i$ to $\text{UE}_{m_k}$, $\sigma^2_0$ is the noise power, and $P_t$ is the total power of each BS.

Denote the precoding policy as $\mathbf{V} = g(\mathbf{\mathbf{H}})$, where $\mathbf{V}=\mathrm{diag}(\mathbf{V}_1,\cdots,\mathbf{V}_M) \in \mathbb{C}^{MN_\mathsf{t} \times MKN_\mathsf{r}}$, $\mathbf{V}_m = [\mathbf{V}_{m_1},\cdots,\mathbf{V}_{m_K}] \in \mathbb{C}^{N_\mathsf{t}\times KN_{\mathsf{r}}}$ is the precoding matrix at $\text{BS}_m$, and $\mathbf{H} = [\mathbf{H}_{i,j_k}] \in \mathbb{C}^{MN_\mathsf{t}\times MKN_{\mathsf{r}}}$ is the channel matrix from all BSs to all UEs.
By using the method in \cite{liu2023multidimensional}, it is not hard to examine that the policy is defined on a set containing $MN_\mathsf{t}$  ANs at all BSs (called $\text{AN}^{\tt BS}$ set for short) and another set containing $MKN_\mathsf{r}$ ANs at all UEs (called $\text{AN}^{\tt UE}$ set for short),  from the dimensions of $\mathbf{V}$ and $\mathbf{\mathbf{H}}$.

The $\text{AN}^{\tt BS}$ set and $\text{AN}^{\tt UE}$ set can either be normal or nested sets, and their relation can either be dependent or independent, relying on the system settings. In the sequel, we consider several special cases of Problem $\mathrm{P1}$ to explain the PE properties induced by the types of sets and their relations.

\subsubsection{MU-MISO Precoding (Ind. 2D-PE Property)}
For a single cell system where the BS serves $K$ single-antenna UEs (i.e., $M=1,N_{\mathsf{r}}=1$), Problem $\mathrm{P1}$ reduces to optimizing precoding in a MU-MISO system, as illustrated in Fig. \ref{fig:graph_MISO}.

 In this system setting, $\mathbf{V} = [\mathbf{v}_1,\cdots,\mathbf{v}_K] \in \mathbb{C}^{N_\mathsf{t} \times K}, \mathbf{H} = [\mathbf{h}_1,\cdots,\mathbf{h}_K] \in \mathbb{C}^{N_\mathsf{t} \times K}$, $\mathbf{v}_k \in \mathbb{C}^{N_{\mathsf{t}}}$ and $\mathbf{h}_k \in \mathbb{C}^{N_{\mathsf{t}}}$ are respectively the precoding and channel vectors of the $k$-th UE (denoted as $\text{UE}_k$).

 Denote the $n$-th AN at the BS as $\text{AN}_n^{\tt BS}$ and denote the AN at $\text{UE}_k$ as $\text{AN}^{\mathtt{UE}_k}$. Since the ANs in the $\text{AN}^{\tt BS}$ set or the $\text{AN}^{\tt UE}$ set (reducing to a UE set because each UE has a single AN in this setting) can be arbitrarily permuted, the problem consists of two normal sets as follows \cite{liu2023multidimensional},

    $\text{AN}^{\tt BS}$ set: \{$\text{AN}_1^{\tt BS}, \cdots,\text{AN}_{N_{\mathsf{t}}}^{\tt BS}$\},

    $\text{AN}^{\tt UE}$ set: \{$\text{AN}^{\mathtt{UE}_1}, \cdots,\text{AN}^{\mathtt{UE}_K}$\} (i.e.,
    \{$\text{UE}_1, \cdots,\text{UE}_K$\}).\\
The MU-MISO precoding policy satisfies the ind. 2D-PE property in \eqref{eq:pre_2d_pe}  \cite{zhao2022understanding}, where $\mathbf{\Pi}_\mathsf{A}$ and $\mathbf{\Pi}_\mathsf{B}$ reflect the permutations of the ANs at the BS and UEs, respectively.

\begin{figure*}[!t] \centering
\subfigure[MU-MISO precoding, $N_{\mathsf{t}}=3, K=2$.] {
 \label{fig:graph_MISO}
\includegraphics[width=0.21\linewidth]{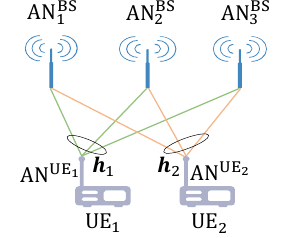}
}\hfil
\subfigure[MU-MIMO precoding, $N_{\mathsf{t}}=3, K=2, N_{\mathsf{r}}=2$.] {
 \label{fig:graph_MIMO}
\includegraphics[width=0.21\linewidth]{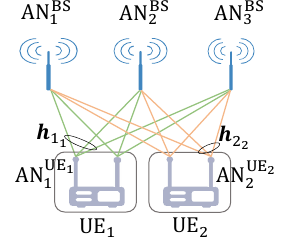}
}
\subfigure[Power control, $M=3$.] {
 \label{fig:graph_PC}
\includegraphics[width=0.21\linewidth]{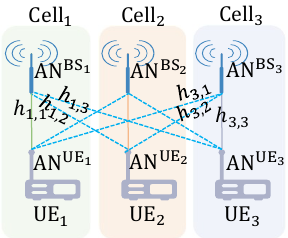}
}
\subfigure[CB and multi-cell power allocation, $N_\mathsf{t}=2,M=2,K=2$.] {
\label{fig:graph_CB}
\includegraphics[width=0.21\linewidth]{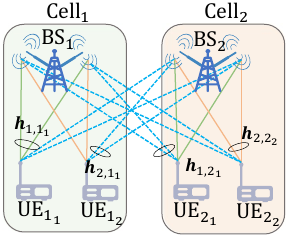}
}\vspace{-2mm}
\caption{Several wireless systems. Boxes indicate the subsets in a nested set. In (c) and (d), solid and dashed lines mean the desired and interference channels.}
\label{fig: graphs}\vspace{-4mm}
\end{figure*}

\subsubsection{MU-MIMO Precoding (Partial-nested 2D-PE Property)}
For a single cell system (i.e., $M=1$) where the BS serves $K$ multi-antenna UEs, Problem $\mathrm{P1}$ reduces to optimizing precoding in a MU-MIMO system, as illustrated in Fig. \ref{fig:graph_MIMO}.

 In this system setting, $\mathbf{H} = [\mathbf{H}_1,\cdots,\mathbf{H}_K]  \in \mathbb{C}^{N_\mathsf{t} \times KN_{\mathsf{r}}}$, where $ \mathbf{H}_k = [\mathbf{h}_{k_1},\cdots,\mathbf{h}_{k_{N_{\mathsf{r}}}}] \in \mathbb{C}^{N_\mathsf{t} \times N_{\mathsf{r}}}$ is the channel matrix of $\text{UE}_k$, and $\mathbf{h}_{k_i} \in \mathbb{C}^{N_\mathsf{t}}$ is the channel vector of the $i$-th antenna of $\text{UE}_k$ (denoted as $\text{AN}_{i}^{\mathtt{UE}_{k}}$). $\mathbf{V} = [\mathbf{V}_1,\cdots,\mathbf{V}_K] \in \mathbb{C}^{N_\mathsf{t} \times KN_{\mathsf{r}}}$, and $\mathbf{V}_k = [\mathbf{v}_{k_1},\cdots,\mathbf{v}_{k_{N_\mathsf{r}}}] \in \mathbb{C}^{N_\mathsf{t} \times N_{\mathsf{r}}}$ is the precoding matrix of $\text{UE}_{k}$.

By examining with the method in \cite{liu2023multidimensional}, the $\text{AN}^{\tt BS}$ set is a normal set and the $\text{AN}^{\tt UE}$ set is a nested set as follows,

$\text{AN}^{\tt BS}$ set: \{$\text{AN}_1^{\tt BS}, \cdots,\text{AN}_{N_{\mathsf{t}}}^{\tt BS}$\},

$\text{AN}^{\tt UE}$ set: $\{\{\text{AN}_1^{\mathtt{UE}_1}\!,\ldots\!,\text{AN}_{N_{\mathsf{r}}}^{\mathtt{UE}_1}\},\ldots\!,\{\text{AN}_1^{\mathtt{UE}_K}\!,\ldots\!,\!\text{AN}_{N_{\mathsf{r}}}^{\mathtt{UE}_K}\}\}$.

Therefore, the MU-MIMO precoding policy has the partial-nested 2D-PE property in \eqref{eq:pre_partial_nested_2d_pe}.

\subsubsection{Power Control (Joint 2D-PE Property)}
When considering $M$ cells where each BS equipped with a single antenna serves only a single-antenna UE (i.e., $N_{\mathsf{t}}=1, N_{\mathsf{r}}=1, K=1$), Problem $\mathrm{P1}$ reduces to a power control problem (say for D2D communications), as illustrated in Fig. \ref{fig:graph_PC}.

In this setting, $\mathbf{H} = [h_{i,j}] \in \mathbb{C}^{M \times M}$, where $h_{i,i}$ is the desired channel from $\text{BS}_i$ to the UE in $\text{Cell}_i$ (denoted as $\text{UE}_i$), and $h_{i,j}\,\,(i\neq j)$ is the interference channel from $\text{BS}_i$ to $\text{UE}_j$. $\mathbf{V} = \mathrm{diag}(v_1,\cdots,v_M) \in \mathbb{R}^{M\times M}$, where $v_i$ represents the transmit power of the $\text{BS}_i$. Denote the AN at $\text{BS}_i$ and $\text{UE}_i$ as $\text{AN}^{\mathtt{BS}_i}$ and $\text{AN}^{\mathtt{UE}_i}$, respectively. The $\text{AN}^{\tt BS}$ set (i.e., the BS set) and $\text{AN}^{\tt UE}$ set  (i.e., the UE set) are constituted as,

    $\text{AN}^{\tt BS}$ set: \{$\text{AN}^{\mathtt{BS}_1}, \cdots,\text{AN}^{\mathtt{BS}_M}$\}, (i.e., \{$\text{BS}_1, \cdots,\text{BS}_M$\}),

    $\text{AN}^{\tt UE}$ set: \{$\text{AN}^{\mathtt{UE}_1}, \cdots,\text{AN}^{\mathtt{UE}_M}$\} (i.e., \{$\text{UE}_1, \cdots,\text{UE}_M$\}).

The ANs in the $\text{AN}^{\tt BS}$ set and $\text{AN}^{\tt UE}$ set should be permuted jointly. Thereby, the power control policy is with the joint 2D-PE property in \eqref{eq:pre_joint_2d_pe} \cite{shenyifei}.

\subsubsection{CB and Power Allocation (Nested 2D-PE Properties)}\label{CBandPC}
When considering $M$ cells where each BS with $N_{\mathsf{t}}$ antennas serves $K$ single-antenna UEs (i.e., $N_{\mathsf{r}}=1$), Problem $\mathrm{P1}$ degenerates into the optimization of CB in multi-cell MU-MISO system, as illustrated in Fig. \ref{fig:graph_CB}, where $\mathbf{h}_{i,j_k} \in \mathbb{C}^{N_{\mathsf{t}}}$ is the channel from $\text{BS}_i$ to $\text{UE}_{j_k}$.

Denote the AN at $\text{UE}_{j_k}$ as $\text{AN}^{\mathtt{UE}_{j_k}}$. By using the method in \cite{liu2023multidimensional}, it can be verified that both $\text{AN}^{\tt BS}$ set and $\text{AN}^{\tt UE}$ set (i.e., the UE set) are nested sets as follows,

\noindent$\text{AN}^{\tt BS}$ set: $\{\!\{\text{AN}_1^{\mathtt{BS}_1}, \!\cdots\!,\text{AN}_{N_{\mathsf{t}}}^{\mathtt{BS}_1}\},\!\cdots\!,\{\text{AN}_1^{\mathtt{BS}_M}, \!\cdots\!,\text{AN}_{N_{\mathsf{t}}}^{\mathtt{BS}_M}\}\!\}$,

\noindent$\text{AN}^{\tt UE}$ set: $\{\!\{\text{AN}^{\tt UE_{1_1}}, \!\cdots\!,\!\text{AN}^{\mathtt{UE}_{1_K}}\},\!\!\cdots\!\!,\{\text{AN}^{\mathtt{UE}_{M_1}},\! \!\cdots\!,\!\text{AN}^{\mathtt{UE}_{M_K}}\}\!\}$ (i.e., $\{\!\{\text{UE}_{1_1}, \!\cdots\!,\!\text{UE}_{1_K}\},\!\!\cdots\!\!,\{\text{UE}_{M_1}, \!\cdots\!,\!\text{UE}_{M_K}\}\!\}$),\\
and the subsets in $\text{AN}^{\tt BS}$ set should be jointly permuted with the subsets in $\text{AN}^{\tt UE}$ set.

Hence, the CB policy has the nested partial-joint 2D-property in \eqref{eq:pre_nested_partial_joint_2d_pe}, where $ \mathbf{\Omega}_{\mathsf{A}}$ and $ \mathbf{\Omega}_{\mathsf{B}}$ correspond to the permutations of the $\text{AN}^{\tt BS}$ set and $\text{AN}^{\tt UE}$ set, $\mathbf{\Pi}_{\mathsf{A,sub}} = \mathbf{\Pi}_{\mathsf{B,sub}}$ jointly permutes the BS and UEs in each cell, and $\mathbf{\Pi}_{\mathsf{A},i}$ and $\mathbf{\Pi}_{\mathsf{B},i}$ permute the ANs at $\text{BS}_i$ and the UEs in $\text{Cell}_{i}$, respectively.

When the channel matrix is replaced by an equivalent channel matrix after beamforming,
Problem $\mathrm{P1}$ degenerates into a multi-cell multi-user power allocation problem.  In this setting, $\mathbf{V}_m = \mathrm{diag}(v_{m_1}, \cdots,v_{m_K}) \in \mathbb{R}^{K \times K}$ is a power allocation matrix at $\text{BS}_m$, where $v_{m_k}$ is the power allocated to $\text{UE}_{m_k}$. This problem also consists of the two nested sets, while both the subsets and the elements in the subsets of the two sets should be jointly permuted. Hence, the power allocation policy satisfies the nested joint 2D-PE property in \eqref{eq:pre_nested_joint_2d_pe}  \cite{guo2022het_pc}.

\vspace{-1mm}\begin{remark}
For SE-maximization MU-MISO precoding, the optimal solution has the following structure \cite{Björnson},
\begin{equation}
\mathbf{V}^*=\left(\mathbf{I}+\frac{1}{\sigma_0^2}\mathbf{H}\boldsymbol{\Lambda}\mathbf{H}^\mathsf{H}\right)^{-1}\mathbf{H}\mathbf{P}^\frac{1}{2}
\label{eq: duality}
\end{equation}
\noindent where $\mathbf{P}=\mathrm{diag}(p_1/||(\mathbf{I}+\frac{1}{\sigma_0^2}\mathbf{H}\boldsymbol{\Lambda}\mathbf{H}^\mathsf{H})^{-1}\mathbf{h}_1||^2,\cdots,$ $p_K/||(\mathbf{I}+\frac{1}{\sigma_0^2}\mathbf{H}\boldsymbol{\Lambda}\mathbf{H}^\mathsf{H})^{-1}\mathbf{h}_K||^2)$, $\boldsymbol{\Lambda}=\mathrm{diag}(\lambda_1,\cdots,\lambda_K)$, $p_k$ and $\lambda_k$ can be viewed as the allocated power to the $\text{UE}_k$, which satisfy $\sum_{k=1}^K\lambda_k=\sum_{k=1}^K{p_k}=P_t, p_k>0,\lambda_k>0, \forall k$.

With this structure, one can first learn the mapping from $\mathbf{H}$ to $\mathbf{P}$ and $\mathbf{\Lambda}$, which is 1D-PE to UEs and 1D-PI to ANs, and then use \eqref{eq: duality} to recover the precoding matrix $\mathbf{V}$. Such a model-driven deep learning has been widely used to optimize the SE-maximization MU-MISO precoding \cite{HybridBeamforming_1D-PEGNN2024, DaiLL_LLMMultitask, TransformerIndirect}.
\end{remark}\vspace{-1mm}

\vspace{-2mm}\begin{remark}
There exist wireless policies with more sets. For example,
consider the precoding policy in wideband MU-MISO systems with multiple resource blocks (RBs). In addition to the $\text{AN}^{\tt BS}$ and $\text{AN}^{\tt UE}$ sets, the RBs also constitute a set \cite{liu2023multidimensional}.
Hence, the wideband MU-MISO precoding policy has the ind. 3D-PE property in \eqref{eq:pre_3d_pe}.
\end{remark}\vspace{-1mm}

\vspace{-3mm}
\section{Recaping 1D-Former and 2D-Former}\label{sec:1D-2D-Former}\vspace{-1mm}
To show how the PE knowledge can be incorporated into Transformer, we recap the 1D- and 2D-Formers designed in \cite{Gformer} for learning the MU-MISO precoding policy, which satisfy the 1D-PE and ind. 2D-PE properties, respectively.

\vspace{-3mm} \subsection{Tokens and Representations}\vspace{-0.5mm}

  The input of  1D- or 2D-Former consists of multiple tokens, each with an associated representation.
The design of tokens and representations is crucial for the learning performance and size-generalizability of Transformer-based DNNs \cite{Gformer}.

 For learning the MU-MISO precoding policy, each UE (say UE$_k$) is designed as a token, whose representation is the channel vector of the UE denoted as $\mathbf{d}_k^{(0)} = [\mathsf{Re}(\mathbf{h}_k^{\mathsf{T}}),\mathsf{Im}(\mathbf{h}_k^{\mathsf{T}})]^{\mathsf{T}} \in \mathbb{R}^{2N_\mathsf{t}}$. In this way, the attention mechanism can reflect MUI, thereby improving the learning performance and enabling the generalizability to the number of UEs.\footnote{If each token is designed as a $\text{AN}^{\tt BS}$, then both the 1D- and 2D-Formers perform worse, and both can be generalized to the number of ANs but are not generalizable to the number of UEs\cite{Gformer}.}

\vspace{-3mm} \subsection{Architecture and Four Modules}\vspace{-0.5mm}
   Both the 1D- and 2D-Formers are with the encoder-only architecture, which omit the decoder and positional encoding.

   The encoder consists of $L$ layers, and each layer is composed of an attention sub-layer (ATT) and a feed-forward network sub-layer (FFN).
   In the $\ell$-th layer, the  representation of the $k$-th token is updated as follows,
\vspace{-2mm}
\begin{align}
    &\textbf{ATT:} \hspace{3mm} \mathbf{c}_{k}^{(\ell)} = \sum_{i=1}^{K} \underbrace{\Big(\mathbf{d}_{k}^{(\ell-1)\mathsf{T}}\big(\mathbf{U}^{\mathsf{K}} \mathbf{d}_{i}^{(\ell-1)}\big)\Big)}_{\alpha_{ki}} \mathbf{U}^{\mathsf{V}}\mathbf{d}_{i}^{(\ell-1)} \label{eq: vanilla-Trans-attention} \\
    &\textbf{FFN:} \hspace{3mm} \mathbf{d}_k^{(\ell)} = \sigma \Big(\mathbf{U}^{\mathsf{F}} \big(\mathbf{d}_{k}^{(\ell-1)} + \mathbf{c}_{k}^{(\ell)}\big) \Big) \label{eq: vanilla-Trans-FFN}
\end{align}
where $\mathbf{d}^{(\ell)}_k \in \mathbb{R}^{J^{(\ell)}N_\mathsf{t}}$, $J^{(\ell)}$ is a hyper-parameter, $\mathbf{U}^{\mathsf{K}}\!,\! \mathbf{U}^{\mathsf{V}} \!\in\! \mathbb{R}^{J^{(\ell-1)}N_\mathsf{t}\times J^{(\ell-1)}N_\mathsf{t}}$ and $\mathbf{U}^{\mathsf{F}}\! \in \!\mathbb{R}^{J^{(\ell)}N_\mathsf{t} \times J^{(\ell-1)}N_\mathsf{t}}$ are trainable weight matrices, and $\sigma(\cdot)$ is an activation function.

$\alpha_{ki} \triangleq \Big(\mathbf{d}_{k}^{(\ell-1)\mathsf{T}}\big(\mathbf{U}^{\mathsf{K}} \mathbf{d}_{i}^{(\ell-1)}\big)\Big)$ is the attention score.
When updating $\mathbf{d}^{(\ell-1)}_k$, the attention scores between $\mathbf{d}^{(\ell-1)}_k$ and all token representations $\mathbf{d}^{(\ell-1)}_i$ ($i=1,...,K$) are computed to capture their dependencies. This mechanism is referred to as \emph{global attention}, which can reflect MUI.

\vspace{-4.5mm}\subsection{Parameter-sharing}\vspace{-1mm}
The 1D- and 2D-Formers differ in the parameter-sharing in order for satisfying different PE properties.

For the 1D-Former, the weight matrices
 are identical (i.e., shared) for all tokens, but these matrices themselves have no structure. In this way, the 1D-Former is equivariant to the permutation of UEs, i.e., satisfying the 1D-PE property in \eqref{eq:pre_1d_pe}. We denote the ATT and FFN with unstructured weight matrices as ``{\bf ATT}",  and ``{\bf FFN}", respectively. Both are modules with 1D-PE property due to the parameter sharing across tokens.

By further sharing trainable parameters based on the 1D-Former, the 2D-Former satisfies the property in \eqref{eq:pre_2d_pe}. Specifically, all the weight matrices in each layer are designed with identical diagonal blocks and identical off-diagonal blocks.
Take $\mathbf{U}^{\mathsf{K}}$ as an example, it has the following structure,\vspace{-1.5mm}

\begin{small}
\begin{equation}\label{eq:matrix_structure}
\mathbf{U}^{\mathsf{K}} = \begin{bmatrix}
    \mathbf{U}^{\mathsf{K}}_1 & \mathbf{U}^{\mathsf{K}}_2 &\mathbf{U}^{\mathsf{K}}_2 \\
    \mathbf{U}^{\mathsf{K}}_2 &\mathbf{U}^{\mathsf{K}}_1 & \mathbf{U}^{\mathsf{K}}_2 \\
    \mathbf{U}^{\mathsf{K}}_2 & \mathbf{U}^{\mathsf{K}}_2 &\mathbf{U}^{\mathsf{K}}_1 \\
\end{bmatrix}
\end{equation}
\end{small}

\vspace{-1mm}\noindent where $\mathbf{U}_1^{\mathsf{K}},\mathbf{U}_2^{\mathsf{K}} \in \mathbb{R}^{J^{(\ell)}\times J^{(\ell)}}$ are trainable parameters whose dimensions are independent of $K$ or $N_\mathsf{t}$.

We denote the ATT and FFN whose weight matrices are with the structure in \eqref{eq:matrix_structure} as ``{\bf ATT (P.S.)}",  and ``{\bf FFN (P.S.)}", respectively, both are modules with the ind. 2D-PE property.

\vspace{-3mm}
\section{Specific PE-Formers and PE-MoFormer}\label{sec:methodology}\vspace{-1mm}
To facilitate the design of cross-task learning, we analyze the hypothesis spaces of DNNs with different PE properties. Then, we design three modules respectively satisfying the properties in \eqref{eq:pre_nested_1d_pe}, \eqref{eq:pre_partial_nested_2d_pe}, and \eqref{eq:pre_joint_2d_pe}. We show that these modules and the four modules in section \ref{sec:1D-2D-Former} can serve as the building blocks, which can be composed into specific PE-Formers for learning the policies in section \ref{w-policies}, each ought to be trained for the policy with a particular PE property. Based on the insight from the design rule of specific PE-Formers, we propose a PE-MoFormer by selecting, assembling, and reusing the modules, which is a single DNN pre-trained for learning multiple wireless policies with matched PE properties.

\vspace{-3mm}\subsection{Hypothesis Spaces of DNNs}\vspace{-0.5mm}
The function family that a DNN (or one of its layers) can represent is called the \emph{hypothesis space} of the DNN (or the layer). For example, the hypothesis space of an FNN is the family of all continuous functions, since FNN can represent any continuous function \cite{universalapp}.

Denote $\mathcal{F}$ as a family of functions that have a type of PE property.
For example, if the input-output relation of a DNN has the 1D-PE property, then the DNN can only represent the functions with the property, i.e., its hypothesis space is $\mathcal{F}_{\text{1D-PE}} \triangleq \{f_\Theta(\cdot)\,|\,f_\Theta(\cdot) \text{ has the 1D-PE property} \}$, where $\Theta$ denotes the trainable parameters in the DNN.

\subsubsection{Hypothesis Spaces of DNNs with Different PE Properties}\label{HYSPDNN}
From \eqref{eq:pre_1d_pe} and \eqref{eq:pre_nested_1d_pe}, we can see that the nested 1D-PE property is a special case of the 1D-PE property. This is because the functions exhibiting the 1D-PE property must satisfy the nested 1D-PE property, due to the fact that the permutation matrix $\mathbf{\Pi}$ in \eqref{eq:pre_1d_pe} includes all possible permutations (including those represented by $\mathbf{\Omega}$ in \eqref{eq:pre_nested_1d_pe}). However, the functions exhibiting the nested 1D-PE property do not necessarily satisfy the 1D-PE property, since the nested 1D-PE property only ensures the equivariance to the permutation of subsets and the elements in each subset.
 Hence, $\mathcal{F}_{\text{1D-PE}} \subsetneq \mathcal{F}_{\text{nested 1D-PE}}$.

By letting $\mathbf{\Pi}_\mathsf{B} = \mathbf{I}$, it can be seen that the 1D-PE property in \eqref{eq:pre_1d_pe} is a special case of the ind. 2D-PE property in \eqref{eq:pre_2d_pe}. This means that the functions satisfying the ind. 2D-PE property must satisfy the 1D-PE property, but the converse statement is not true. Hence, $\mathcal{F}_{\text{ind. 2D-PE}} \subsetneq \mathcal{F}_{\text{1D-PE}}$.

Using similar analyses, the relation among the hypothesis spaces of the DNNs satisfying the PE properties in section \ref{sec:PEs} can be summarized as follows, which is shown in Fig. \ref{fig:hypo}.

\begin{enumerate}
\item  $\mathcal{F}_{\text{ind. 3D-PE}} \subsetneq \mathcal{F}_{\text{ind. 2D-PE}} \subsetneq \mathcal{F}_{\text{partial-nested 2D-PE}} \\
    \subsetneq \mathcal{F}_{\text{nested ind. 2D-PE}} \subsetneq \mathcal{F}_{\text{1D-PE}} \subsetneq \mathcal{F}_{\text{nested 1D-PE}}$

\item  $\mathcal{F}_{\text{ind. 3D-PE}} \subsetneq \mathcal{F}_{\text{ind. 2D-PE}} \subsetneq \mathcal{F}_{\text{joint 2D-PE}} \\
    \subsetneq \mathcal{F}_{\text{nested partial-joint 2D-PE}} \subsetneq \mathcal{F}_{\text{nested joint 2D-PE}}$

\item $\mathcal{F}_{\text{nested ind. 2D-PE}} \subsetneq \mathcal{F}_{\text{nested partial-joint 2D-PE}}$

\end{enumerate}

\begin{figure}
    \centering
    \includegraphics[width=0.75\linewidth]{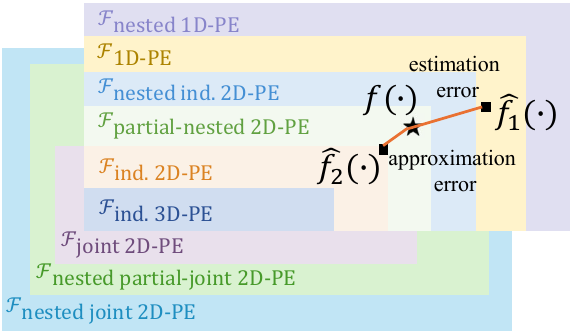}\vspace{-3mm}
    \caption{Relation of hypothesis spaces of DNNs with different PE properties.}
    \label{fig:hypo}\vspace{-4mm}
\end{figure}

When designing a DNN to learn a function with a type of PE property, appropriately constraining its hypothesis space by exploiting the property can enhance its learning performance and learning efficiency \cite{zhao2022understanding}.

For instance, when learning a function $f(\cdot)$ with the partial-nested 2D-PE property, an efficient DNN should be designed such that its hypothesis space is $\mathcal{F}_{\text{partial-nested 2D-PE}}$. If the hypothesis space of the DNN  (e.g., $ \mathcal{F}_{\text{1D-PE}} $) is \emph{larger} than $ \mathcal{F}_{\text{partial-nested 2D-PE}} $, then more training samples are required to well-approximate the function. With limited samples, the learning performance may degrade due to the \emph{estimation error} \cite{zhao2022understanding}, which is ``distance" between $f(\cdot)$ and $\hat{f_1}(\cdot)$ in Fig. \ref{fig:hypo}. Conversely, if the hypothesis space  (e.g., $ \mathcal{F}_{\text{ind. 2D-PE}} $) is \emph{smaller} than $ \mathcal{F}_{\text{partial-nested 2D-PE}} $ such that $f(\cdot)$ is not in the hypothesis space, the DNN can only learn a function as close as possible to $f(\cdot)$ in $\mathcal{F}_{\text{ind. 2D-PE}}$ (i.e., $\hat{f}_2(\cdot)$ in Fig. \ref{fig:hypo}). Then, the approximation error (i.e., the ``distance" between $f(\cdot)$ and  $\hat{f}_2(\cdot)$) has to be mitigated by enlarging the DNN scale (i.e., using a wider or deeper DNN). Both kinds of errors degrade the performance or the learning efficiency of DNNs.

\subsubsection{Hypothesis Space of a DNN Consisting of Layers with Different PE Properties}\label{HYSPDNNL}
When the layers in a DNN satisfy different PE properties (i.e., have different hypothesis spaces), \emph{the hypothesis space of the DNN depends on the layer with the largest hypothesis space}. Formally, consider a two-layer DNN whose input-output relations of the first and second layers are $g_{\Theta_1}(\cdot) \in \mathcal{F}_1$, $f_{\Theta_2}(\cdot) \in \mathcal{F}_2$, where $\mathcal{F}_1$ and $\mathcal{F}_2$ are respectively the hypothesis spaces of the two layers.

\begin{proposition}\label{prop:compound_func}
    If $\mathcal{F}_1$ is $\mathcal{F}_{\text{ind. 2D-PE}}$ and $\mathcal{F}_1 \subsetneq \mathcal{F}_2$, then $ f_{\Theta_2}(g_{\Theta_1}(\cdot)) \in \mathcal{F}_2 $ but $ f_{\Theta_2}(g_{\Theta_1}(\cdot)) \in \mathcal{F}_1 $ does not necessarily hold. Conversely, if $\mathcal{F}_2$ is $\mathcal{F}_{\text{ind. 2D-PE}}$ and $\mathcal{F}_2 \subsetneq \mathcal{F}_1$, $ f_{\Theta_2}(g_{\Theta_1}(\cdot)) \in \mathcal{F}_1 $ but $ f_{\Theta_2}(g_{\Theta_1}(\cdot)) \in \mathcal{F}_2 $ may not always hold.
\end{proposition}

\vspace{-2mm}\begin{proof}
    See Appendix~\ref{proof:compound_func}.
\end{proof}
We can prove that Proposition \ref{prop:compound_func} still holds without the assumption of $\mathcal{F}_1$ (or $\mathcal{F}_2$) being $\mathcal{F}_{\text{ind. 2D-PE}}$, by using similar derivations as in Appendix~\ref{proof:compound_func}.

\vspace{-2mm}\subsection{Design of Three Modules}\vspace{-1mm}
From the analyses in section \ref{HYSPDNN}, we know that when learning policies with joint PE and nested PE properties in \eqref{eq:pre_nested_1d_pe}, \eqref{eq:pre_joint_2d_pe}-\eqref{eq:pre_nested_joint_2d_pe}, the 1D-Former and 2D-Former either requires a large number of training samples for reducing the estimation error or is hard to achieve satisfactory performance due to approximation error.
In order to design DNNs with matched PE properties to such policies, we propose another three modules in the sequel.

\subsubsection{Nested Attention Sub-layers}
 We first design a nested attention sub-layer (NATT) of the $\ell$-th layer that satisfies the nested 1D-PE property in \eqref{eq:pre_nested_1d_pe}, which is a module denoted as {\bf ``NATT''}. Then, we design the structure of the weight matrices in the sub-layer, such that the NATT with parameter-sharing can satisfy the partial-nested 2D-PE property in \eqref{eq:pre_partial_nested_2d_pe}, which is another module denoted as {\bf ``NATT (P.S.)}".

 Consider a nested set with $N_{\mathsf{sub}}$ subsets where each subset consists of $N_{\mathsf{s}}$ elements. We design each token as an element in the nested set, and denote the token representation of the $n$-th element in the $k$-th subset as $\mathbf{d}_{k_n}^{(\ell)}$.

The key to designing a module with the nested 1D-PE property lies in distinguishing the tokens within the same subset and across different subsets. If representations of all the tokens in the nested set are updated in the same way (say by using the global attention mechanism as in ATT), the input-output relation of the sub-layer will be equivariant to the permutation of elements not only in the same subset but also across different subsets. Then, the hypothesis space of the sub-layer is  $\mathcal{F}_{\text{1D-PE}}$, which is smaller than $\mathcal{F}_{\text{nested 1D-PE}}$.

To differentiate the tokens within the same subset and across different subsets, one can apply an existing approach of designing parameter-sharing for weight matrices \cite{liu2023multidimensional}. However, unlike the ind. 2D-PE property, designing parameter-sharing for satisfying the nested PE property is much more tedious.

Instead, we propose a NATT with both local and global attention mechanisms. The \emph{local attention} only computes the scores of the token representations in the same subset (e.g., $\mathbf{d}_{k_n}^{(\ell-1)}$ and $\mathbf{d}_{k_i}^{(\ell-1)},i=1,\cdots,N_{\mathsf{s}}$). The global attention computes the scores of all token representations. In this way, the local attention mechanism can distinguish tokens in the same subset from different subsets, and the global attention mechanism can extract the information from different subsets.

Specifically, the representation of the $n$-th token in the $k$-th subset of the $\ell$-th layer is updated as follows,\vspace{-2mm}
 \begin{equation}\label{eq:2d-nested-att}
\begin{aligned}
\textbf{NATT:} \hspace{2mm} \mathbf{c}_{k_n}^{(\ell)} &=  \underbrace{\sum_{i=1}^{N_{\mathsf{s}}} \Big(\mathbf{d}_{k_n}^{(\ell-1){\mathsf{T}}}
		\big(\mathbf{U}^{\mathsf{K}}_{\mathsf{S}} \mathbf{d}_{k_i}^{(\ell-1)}\big)\Big)\mathbf{U}^{\mathsf{V}}_\mathsf{S}\mathbf{d}_{k_i}^{(\ell-1)}}_{\mathrm{local\,\, attention  \,\, }}+\\&\quad\underbrace{\sum_{j=1}^{N_{\mathsf{sub}}}\sum_{i=1}^{N_{\mathsf{s}}} \Big(\mathbf{d}_{k_n}^{(\ell-1){\mathsf{T}}}
		\big(\mathbf{U}^{\mathsf{K}}_{\mathsf{D}} \mathbf{d}_{j_i}^{(\ell-1)}\big)\Big)\mathbf{U}^{\mathsf{V}}_\mathsf{D}\mathbf{d}_{j_i}^{(\ell-1)}}_{\mathrm{global \,\, attention}}
\end{aligned}
\end{equation}
where $\mathbf{U}^{\mathsf{K}}_{\mathsf{S}},\mathbf{U}^{\mathsf{V}}_{\mathsf{S}},\mathbf{U}^{\mathsf{K}}_{\mathsf{D}},\mathbf{U}^{\mathsf{V}}_{\mathsf{D}}$ are trainable weight matrices.

\begin{proposition}\label{prop:nested-pe}
    The input-output relation of the NATT satisfies the nested 1D-PE property in \eqref{eq:pre_nested_1d_pe}.
\end{proposition}
\vspace{-2mm}\begin{proof}
    See Appendix \ref{proof:nested-pe}
\end{proof}


\vspace{-1mm}\begin{proposition}\label{prop:sharing_params}
    When $\mathbf{\Pi}\mathbf{U}^{\mathsf{K}}_{\mathsf{S}} = \mathbf{U}^{\mathsf{K}}_{\mathsf{S}}\mathbf{\Pi}, \mathbf{\Pi}\mathbf{U}^{\mathsf{V}}_{\mathsf{S}} = \mathbf{U}^{\mathsf{V}}_{\mathsf{S}}\mathbf{\Pi},\mathbf{\Pi}\mathbf{U}^{\mathsf{K}}_{\mathsf{D}} = \mathbf{U}^{\mathsf{K}}_{\mathsf{D}}\mathbf{\Pi},\mathbf{\Pi}\mathbf{U}^{\mathsf{V}}_{\mathsf{D}} = \mathbf{U}^{\mathsf{V}}_{\mathsf{D}}\mathbf{\Pi}$, the input-output relation of the {\bf NATT (P.S.)} satisfies the partial-nested 2D-PE property  in \eqref{eq:pre_partial_nested_2d_pe}.
\end{proposition}
\vspace{-2mm}\begin{proof}
    See Appendix \ref{proof:sharing_params}
\end{proof}
If the matrices $\mathbf{U}^{\mathsf{K}}_{\mathsf{S}},\mathbf{U}^{\mathsf{V}}_{\mathsf{S}},\mathbf{U}^{\mathsf{K}}_{\mathsf{D}},\mathbf{U}^{\mathsf{V}}_{\mathsf{D}}$ are with the structure in \eqref{eq:matrix_structure}, the four conditions in the proposition can satisfied \cite{zaheer2017deep}.

\vspace{-1mm}\begin{remark}
 In \cite{li2024hpe}, a Hier-Former was proposed for multi-cast beamforming in MU-MISO systems,
 which adopts an encoder-decoder architecture. The encoder learns power allocation, which is recovered to the beamforming matrix by the decoder using a structure similar to \eqref{eq: duality}.
However, the Hier-Former can only satisfy the nested 1D-PE property but not the partial-nested 2D-PE property satisfied by the NATT (P.S.).


\end{remark}\vspace{-2mm}

\subsubsection{Diagonal Output Layer}\label{subsec:joint and nested PE}
In \eqref{eq:pre_joint_2d_pe}, the rows and columns of $\mathbf{Y}$ and $\mathbf{X}$ are permuted with the same matrix. It is easy to show that for any matrix $\mathbf{X}$ and any permutation matrix $\mathbf{\Pi}$, the diagonal elements of $\mathbf{\Pi}^{\mathsf{T}}\mathbf{X}\mathbf{\Pi}$ are the same as $\mathbf{X}$ but with permuted order.\footnote{For example, let $\mathbf{X} = \begin{bmatrix} a & b \\ c & d \end{bmatrix}$ and $\boldsymbol{\Pi} = \begin{bmatrix} 0 & 1 \\ 1 & 0 \end{bmatrix}$.
$\mathbf{X}$ and $\boldsymbol{\Pi} \mathbf{X} \boldsymbol{\Pi}^\mathsf{T} = \begin{bmatrix} d & c \\ b & a \end{bmatrix}$ have the same diagonal elements (i.e., $a,d$) with a reversed order.} This is not true for other PE properties such as the ind. 2D-PE property, where the diagonal elements of $\mathbf{\Pi}_{\mathsf{A}}^{\mathsf{T}}\mathbf{X}\mathbf{\Pi}_{\mathsf{B}}$ and $\mathbf{X}$ may  completely differ.

The observation suggests that the following function extracting the diagonal elements satisfies the joint 2D-PE property,\vspace{-1mm}
\begin{equation}
    \mathbf{Y} = f_{\mathsf{Diag}}(\mathbf{X}), \quad \text{where} \quad y_{i,j} =
    \begin{cases}
    x_{i,j}, & \text{if } i = j \\
    0, & \text{if } i \neq j
    \end{cases}
\end{equation}
It is not hard to prove that
$\mathbf{\Pi}^\mathsf{T}_\mathsf{A}\mathbf{Y}\mathbf{\Pi}_{\mathsf{B}} = f_{\mathsf{Diag}}(\mathbf{\Pi}_{\mathsf{A}}^\mathsf{T}\mathbf{X}\mathbf{\Pi}_{\mathsf{B}})$ if and only if $\mathbf{\Pi}_{\mathsf{A}} = \mathbf{\Pi}_{\mathsf{B}}$, indicating that the function $f_{\mathsf{Diag}}(\cdot)$ possesses the joint 2D-PE property.

$f_{\mathsf{Diag}}(\cdot)$ lacks trainable parameters and hence cannot be used alone. We can introduce trainable parameters by combining the function with other layers. For example, according to Proposition \ref{prop:compound_func}, if $g_{\Theta}(\cdot)$ satisfies the ind. 2D-PE property, then $f_{\mathsf{Diag}}(g_{\Theta}(\cdot))$ satisfies the joint 2D-PE property because $\mathcal{F}_{\text{ind. 2D-PE}} \subsetneq \mathcal{F}_{\text{joint 2D-PE}}$ (see Fig. \ref{fig:hypo}). Hence, we can use a DNN that satisfies the ind. 2D-PE property as $g_{\Theta}(\cdot)$ and use $f_{\mathsf{Diag}}(\cdot)$ as an output layer. Then, $f_{\mathsf{Diag}}(g_{\Theta}(\cdot))$ satisfies the property in \eqref{eq:pre_joint_2d_pe} and is trainable.
$f_{\mathsf{Diag}}(\cdot)$ is a module that can be used at the output layer\footnote{Using $f_{\mathsf{Diag}}(\cdot)$ as input layer also ensures the joint 2D-PE property, which however incurs the loss of information in the off-diagonal elements.} of a DNN to satisfy the PE properties with ``joint'' (say \eqref{eq:pre_joint_2d_pe},  \eqref{eq:pre_nested_partial_joint_2d_pe} and \eqref{eq:pre_nested_joint_2d_pe}).

\vspace{-2.9mm}\subsection{Design of Specific PE-Formers}\label{subsec:specific PE-Former}\vspace{-1mm}
From the analyses in section \ref{HYSPDNNL}, we can develop a specific PE-Former to match the PE property of a wireless policy by integrating previous modules and carefully designing tokens and their representations. To showcase the practical applicability, we first take the learning of several policies in section \ref{w-policies} as examples and then summarize the design rule.

\subsubsection{MU-MIMO Precoding (Partial-nested 2D-PE Property)}
 To allow the attention mechanism for modeling the interference between UEs and data streams, $\text{AN}^{\mathtt{UE}_k}_n$ is defined as a token, whose representation is its channel vector, i.e., $\mathbf{d}_{k_n}^{(0)} = [\mathsf{Re}(\mathbf{h}_{k_n}^{\mathsf{T}}),\mathsf{Im}(\mathbf{h}_{k_n}^{\mathsf{T}})]^{\mathsf{T}}$, similar to the 2D-Former.

From Fig. \ref{fig:hypo}, we know that NATT (P.S.) cascaded by FFN (P.S.) satisfies the partial-nested 2D-PE property. According to Proposition \ref{prop:compound_func}, we can compose the two modules in each layer to learn the MU-MIMO precoding policy.

In the $\ell$-th layer, the token representation is first updated by NATT (P.S.) with $\eqref{eq:2d-nested-att}$, where $N_{\mathsf{s}}=N_{\mathsf{r}}, N_{\mathsf{sub}} = K$, and $\mathbf{U}^{\mathsf{K}}_{\mathsf{S}},\mathbf{U}^{\mathsf{V}}_{\mathsf{S}},\mathbf{U}^{\mathsf{K}}_{\mathsf{D}},\mathbf{U}^{\mathsf{V}}_{\mathsf{D}}$ are with the structure in \eqref{eq:matrix_structure}, which is then input into FFN (P.S.). After updated by $L$ layers, the output $\mathbf{d}^{(L)}_{k_n}$
is normalized as $\mathbf{v}_{k_n} = \mathbf{d}_{k_n}^{(L)}  \sqrt{P_t / \sum_{i=1}^{K}\sum_{j=1}^{N_{\mathsf{r}}}  \mathbf{d}_{i_j}^{(L)\mathsf{T}}\mathbf{d}_{i_j}^{(L)}}$ to meet the power constraint. In NATT (P.S.), the local attention and global attention respectively reflect the inter-stream interference of a UE and the interference among UEs.

\subsubsection{CB (Nested Partial-joint 2D-PE Property)}  To satisfy the PE property of the CB policy, NATT (P.S.), FFN (P.S.), and the diagonal output layer are required. Both the $\text{AN}^{\tt BS}$ set and $\text{AN}^{\tt UE}$ set are nested sets in this setting, while using NATT (P.S.) alone can only satisfy the PE property for a single nested set. Yet interference only exists among users but not among antennas, thereby the attention mechanism only needs to be introduced into the $\text{AN}^{\tt UE}$ set  \cite{Gformer}.


By designing the AN at each UE (say UE$_{m_k}$) as a token, NATT (P.S.) can satisfy the nested PE property induced by $\text{AN}^{\tt UE}$ set (i.e., the UE set), meanwhile the attention mechanism can reflect MUI and inter-cell interference (ICI). To satisfy the nested PE property induced by the $\text{AN}^{\tt BS}$ set, we distinguish $\text{ANs}^{\tt BS}$ across different subsets via designing different token representations. For each token, we design the channel vector from each BS (say $\text{BS}_{m'}$) to UE$_{m_k}$ as a representation, i.e., $\mathbf{d}^{(0)}_{m',m_k} = [\mathsf{Re}(\mathbf{h}_{m',m_k}^{\mathsf{T}}), \mathsf{Im}(\mathbf{h}_{m',m_k}^{\mathsf{T}})]^\mathsf{T}$. Since there are $M$ BSs, each token is associated with $M$ representations.

 In the $\ell$-th layer, the $m'$-th ($m'=1,\cdots, M$) representation of the $m_k$-th token is first updated by NATT (P.S.) with \eqref{eq:2d-nested-att}, where $N_{\mathsf{s}}=K, N_{\mathsf{sub}} = M$,
and $\mathbf{U}^{\mathsf{K}}_{\mathsf{S}},\mathbf{U}^{\mathsf{V}}_{\mathsf{S}},\mathbf{U}^{\mathsf{K}}_{\mathsf{D}},\mathbf{U}^{\mathsf{V}}_{\mathsf{D}}$ are with the structure in \eqref{eq:matrix_structure}. The representations are then updated by FFN (P.S.).
After $L$ layers, the output $\mathbf{D}^{(L)}=[\mathbf{d}_{m',m_k}^{(L)}]$ is mapped to  $\mathbf{V} = f_{\mathsf{Diag}}(\mathbf{D}^{(L)})$ by the diagonal output layer. $\mathbf{V}$ is then normalized to satisfy the power constraint at each BS. The local and global attention reflect MUI and ICI, respectively.

\subsubsection{Power Allocation (Nested Joint 2D-PE Property)}
The only difference in the PE property of the power allocation policy in section \ref{CBandPC} from that of the CB policy is that the elements in each subset ought to be permuted jointly instead of independently. We can add additional diagonal output layers on the PE-former for CB to satisfy its property. Specifically, the updated representation after $L$ layers each consisting of NATT (P.S.) and FFN (P.S.), $\mathbf{D}^{(L)} \in \mathbb{R}^{MK\times MK}$, is first mapped to $\mathbf{V}' = f_{\mathsf{Diag}}(\mathbf{D}^{(L)}) \in \mathbb{R}^{MK\times MK}$, where $\mathbf{V}' \triangleq \mathrm{diag}(\mathbf{V}_1',\cdots,\mathbf{V}_M') $, $\mathbf{V}_m' \in \mathbb{R}^{K\times K}$. Then, for the representation of tokens within each subset (say the $m$-th subset), we apply an additional diagonal output layer as $\mathbf{V}_m = f_{\mathsf{Diag}}(\mathbf{V}_m')$ for satisfying the nested joint PE property.

\subsubsection{Wideband MISO Precoding (Ind. 3D-PE Property)}
For learning this policy, each UE is defined as a token whose representation is its channel vector. The ATT (P.S.) and FFN (P.S.) should be used to satisfy the PE property induced by $\text{AN}^{\tt BS}$ set and UE set. To further satisfy the PE property induced by the $\text{RB}$ set, both modules are shared among RBs.

\subsubsection{Rule of Module Composition and Token Design} Specific PE-formers can be designed with the following principle.
\begin{enumerate}
    \item  Select modules with appropriate hypothesis spaces according to Proposition \ref{prop:compound_func} and Fig. \ref{fig:hypo}.
    \item  Design tokens such that the attention can model the correlation of channels concerned in a task.
    \item  Design token representations to distinguish tokens across different subsets when more than one nested sets exist.
\end{enumerate}
We list ten PE-Formers with different PE properties by composing seven modules in different colors in Fig. \ref{fig:diff-mod-comb}, and summarize the rule for the design as follows.

\begin{enumerate}
    \item \emph{1D-PE or nested 1D-PE property}: Modules without parameter-sharing are composed. For 1D-PE or nested 1D-PE, ATT or NATT is selected, cascaded by an FFN.
    \item \emph{2D-PE properties without joint permutation}: Modules with parameter-sharing are composed (i.e., NATT (P.S.) or ATT (P.S.), and FFN (P.S.)). If there is no nested set, ATT (P.S.) is selected. Otherwise, NATT (P.S.) is selected. If both sets are nested sets, token representations should be designed as for CB.
    \item \emph{2D-PE properties with joint permutation:} For joint 2D-PE or nested partial-joint 2D-PE, a diagonal output layer is added. For nested joint 2D-PE, multiple diagonal output layers are required.
    \item \emph{Ind. and partial-nested 3D-PE properties:} The modules adopted for satisfying 2D-PE properties are shared on the third set (e.g., $\text{RB}$ set).
\end{enumerate}

\begin{figure}
    \centering

    \subfigure[Rule of module composition for specific PE-Formers]{
    \includegraphics[width=0.9\linewidth]{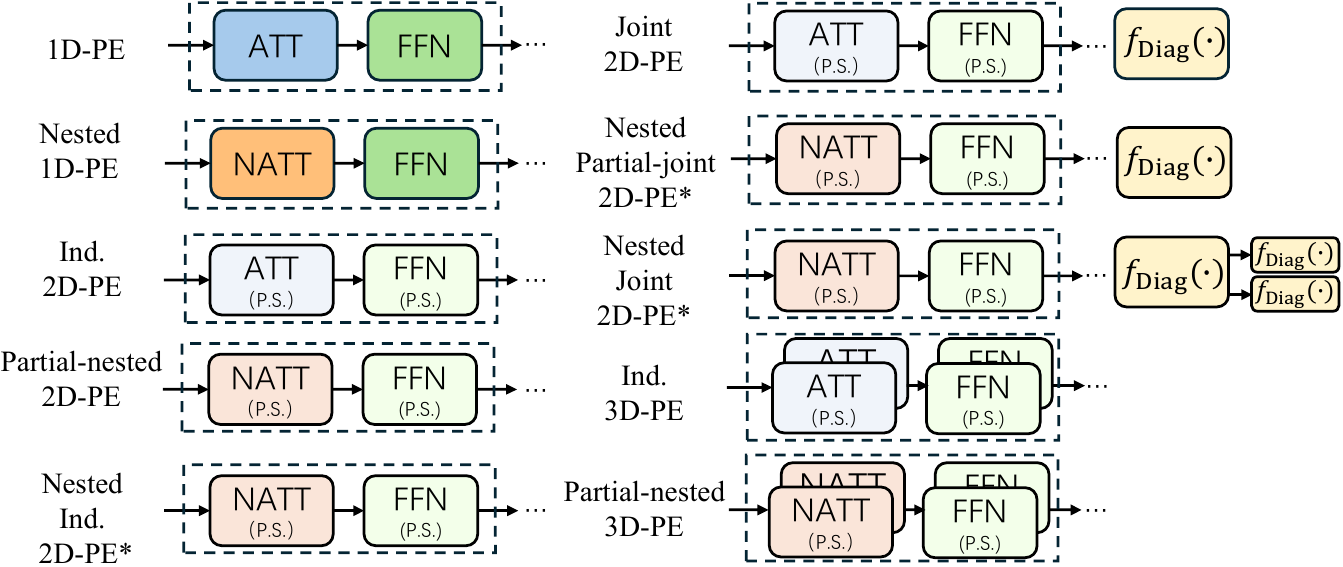}\label{fig:diff-mod-comb}}

    \subfigure[PE-MoFormer]{
        \includegraphics[width=0.8\linewidth]{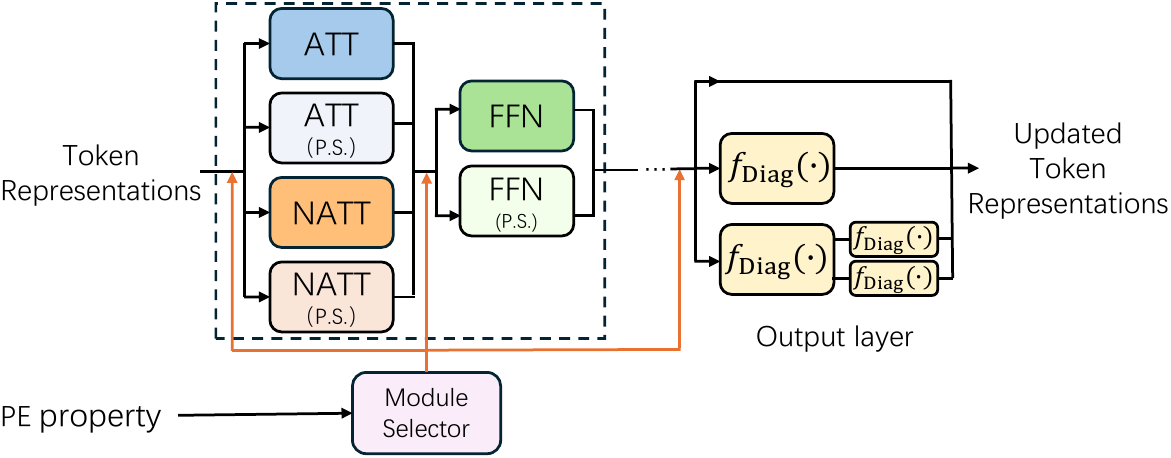}\label{fig:PE-Formerframework}
    }\vspace{-2mm}
    \caption{Architectures of specific PE-Formers and PE-MoFormer.
    Dashed boxes indicate a layer. For satisfying PE properties induced by two nested sets (marked with ``*"), each token has multiple representations.}
    \label{}\vspace{-6mm}
\end{figure}


\vspace{-3mm}
\subsection{Design of PE-MoFormer}\label{subsec:universal PE-Former}\vspace{-0.5mm}
As shown in Fig. \ref{fig:diff-mod-comb}, different PE properties can be satisfied by reusing several modules. Meanwhile, many wireless policies share common structural properties. Inspired by these observations, we devise a \textbf{\emph{PE-MoFormer}} as illustrated in Fig. \ref{fig:PE-Formerframework}, which assembles the modules to accommodate diverse PE properties.
During inference, only the selected modules are activated for a given wireless policy, hence the computational complexity is the same as a specific PE-Former.

\subsubsection{Architecture}
The PE-MoFormer consists of $L$ layers and a module selector, each layer is composed of several cascaded modules selected from three types of modules.

One type of modules includes the attention sub-layers, i.e., {\bf ATT}, {\bf ATT (P.S.)}, {\bf NATT}, and {\bf NATT (P.S.)}. The second type of modules includes {\bf FFN} and {\bf FFN (P.S.)}. The third type of module is the {\bf diagonal output layer} (i.e., $f_{\mathsf{Diag}}(\cdot)$).\footnote{For non-PE policies, the positional encoding can be added on the token representations, and then input into layers composed of {\bf ATT} and {\bf FFN}.}

The module selector stores two tables. One stores the PE properties of policies as exemplified in Table \ref{tab:PE property of wireless policies}, and the other
stores the module composition rule illustrated in Fig. \ref{fig:diff-mod-comb}.

\begin{table}[htbp]
  \centering
 \caption{Over 25 Wireless policies with 11 types of PE properties}
 \vspace{-2mm}
    \begin{tabular}{c|c}
    \hline\hline
    Wireless policies      & PE properties  \\ \hline
    \makecell{SU-MISO channel prediction, estimation \cite{CPdai}} & Non-PE
    \\ \hline

    \makecell{MU-MIMO / Wideband SU-MISO channel \\ prediction,  estimation \cite{LinFormer,AiBo_Tcom_PlugPlay,wirelessGPT} \\ Bandwidth allocation \cite{mehrabian2024joint,Satellite_1D-PEGNN2024}} & 1D-PE \\ \hline
    Multi-cast beamforming \cite{li2024hpe} & Nested 1D-PE  \\\hline
    \makecell{MU-MISO precoding \cite{zhao2022understanding,GAT,Kansformer} \\ User association \cite{userassoc} \\Signal detection \cite{DaiLL_LLMMultitask,Detection} } & Ind. 2D-PE \\ \hline
    \makecell{MU-MIMO precoding \cite{MIMOSVD}\\ Cell-free / Cooperative MU-MISO \\ beamforming \cite{ENGNNLiyang_2024, cell_freeMISO}} & Partial-nested 2D-PE \\ \hline
    \makecell{D2D power control \cite{shenyifei}, link scheduling \cite{Liye}} & Joint 2D-PE \\ \hline
    \makecell{Cell-free / Cooperative MU-MIMO \\ beamforming \cite{cell_freeMIMO}} & Nested ind. 2D-PE \\
    \hline
    \makecell{CB \cite{MISOCB}} & Nested partial-joint 2D-PE \\
    \hline
    \makecell{ Multi-cell power allocation \cite{guo2022het_pc}} & Nested joint 2D-PE
    \\ \hline
    \makecell{IRS-aided / Wideband MU-MISO \\ precoding \cite{ZBCRIS,MIMOSVD}\\ Hybrid MU-MISO precoding \cite{Gformer,AdaTTD_NearField, liu2023multidimensional}\\ Joint beamforming and sensing \cite{ISAC}} & Ind. 3D-PE \\ \hline

    \makecell{IRS-aided / Wideband
    MU-MIMO \\ precoding \cite{mehrabian2024joint,MIMOSVD} \\ Partially-connected hybrid precoding \cite{HybridBeamforming_1D-PEGNN2024}} & Parital-nested 3D-PE \\ \hline

    \hline\hline
    \end{tabular}%
  \label{tab:PE property of wireless policies} \vspace{-3mm}
\end{table}%

\subsubsection{Training Strategies}\label{sec:training strategies}
One training strategy is multi-task learning, where the learning of each wireless policy is treated as a task, and the training set consists of the samples from multiple tasks. The loss function is defined as $ \mathcal{L} = \sum_{i} \lambda_i\mathcal{L}_i$ for joint training,
where $\mathcal{L}_i$ and $\lambda_i$ are the loss and weight of the $i$-th task, respectively.

Multi-task learning may lead to conflicts among multiple tasks, degrading the performance of some tasks. An alternative strategy is cross-task learning \cite{Yuwei_Transferlearning_2024}, where the PE-MoFormer is first pre-trained on some tasks and then fine-tuned for a specific task unseen during pre-training.

\vspace{-1mm}\begin{remark}
The PE-MoFormer is adaptable to various objective functions. For instance, it can be pre-trained to maximize the sum rate and then fine-tuned to maximize the minimum rate (only with 100 samples according to our simulation).
\end{remark}

\vspace{-3mm}
\section{Simulations}\label{sec:simulations}\vspace{-1mm}
In this section, we first evaluate the learning performance and generalizability of specific PE-Formers for several wireless policies to show the impact of the proper composition of designed modules. Then, we evaluate the performance of the PE-MoFormer and compare it with an LLM-based DNN.

\vspace{-3.5mm}\subsection{Simulation Setup and Hyper-parameters}\vspace{-1mm}
The number of ANs at a BS is $N_{\mathsf{t}}=64$, the number of UEs in a cell is $K=8$, and SNR = $10$ dB.
All samples for precoding are generated from the Saleh-Valenzuela (SV) channel model with four clusters and five rays. For wideband MU-MISO precoding, the samples are generated from a tap delay-$d$ SV channel model, where the number of taps is set as two, and the number of RBs is set as eight. For MU-MIMO precoding, the number of ANs at each UE is $N_{\mathsf{r}}=4$. For CB, the number of cells is $M=3$, and the ratio of desired channel gains to interference channel gains is randomly selected from $(0.5,1)$ since the interference channels are generally weaker than the desired channels.
For multi-cell power allocation, the samples are generated from the equivalent channel after zero-forcing beamforming, where other setups are the same as those for CB. For each task, 1,000 samples are generated for testing.

The hyper-parameters of each DNN are listed in Table \ref{tab:hyper_params}, and the activation function in hidden layers is $\mathrm{tanh(x)} = \frac{e^x -e^{-x}}{e^x + e^{-x}}$.
These setups will be used unless otherwise specified.

  \vspace{-2mm}\begin{table}[!ht]
    \centering
    \caption{Fine-tuned Hyper-parameters}
    \vspace{-2mm}
    \begin{tabular}{c|c|c|c}
    \hline\hline
        DNNs & Layer & $J^{(\ell)}$ in hidden layers & Learning rate \\ \hline
        \makecell{Specific PE-Formers/\\ PE-MoFormer
        } & 3 & [32,32,32] & 0.002 \\ \hline
        2D-Former & 3 & [32,32,32] & 0.002 \\ \hline
        Edge-GCN & 4 & [64,64,64,64] & 0.002 \\ \hline
        Hier-Former&3 &[16,16,16] & 0.0005\\ \hline\hline
    \end{tabular}\label{tab:hyper_params}    \vspace{-4mm}
\end{table}

\vspace{-3mm}\subsection{Specific PE-Formers: Impact of Modules and Composition}\vspace{-1mm}
We select two tasks: MU-MIMO precoding and CB, to evaluate the specific PE-Formers (with legend ``PE-Former''). The SE maximization problems under power constraint are considered. The learning performance is measured by the SE ratio, i.e., the ratio of the SE achieved by a policy learned by a DNN to the SE achieved by the weighted minimal mean squared error (WMMSE) algorithm \cite{wmmse}. All the DNNs are trained in an unsupervised manner with an Adam optimizer. The loss function is the negative SE averaged over all training samples.
We compare with the following DNNs.
\begin{itemize}
    \item \textbf{Edge-GCN:} This is a GCN framework in \cite{liu2023multidimensional}, where the parameter-sharing of weight matrices is designed to satisfy the PE properties of the considered policies.
    \item \textbf{2D-Former:} This is a specific PE-Former in \cite{Gformer}, which only satisfies the  ind. 2D-PE property.
    \item \textbf{Hier-Former:} This is a revised version of the DNN proposed in \cite{li2024hpe} for learning MU-MIMO precoding or CB policies, where the token is designed as a $\text{AN}^{\tt UE}$, the decoder is removed, and the PE property induced by the $\text{AN}^{\tt BS}$ set is neglected.
\end{itemize}

\vspace{-0.1mm}\subsubsection{Learning Performance}
In Figs. \ref{fig:MU-MIMO performance} and \ref{fig:CB performance}, we show the SE ratio of each DNN for learning the two policies. We can see that specific PE-Formers can achieve satisfactory performance with a small number of training samples (achieving 95\% performance with fewer than 40 samples). By contrast, the 2D-Former cannot achieve the same performance even with 40,000 training samples. This is because the hypothesis space of the 2D-Former is smaller than $\mathcal{F}_{\text{partial-nested PE}}$, which leads to the performance loss incurred by approximation error. The Hier-Former can achieve 95\% performance with 40,000 training samples for learning MU-MIMO precoding and 10,000 samples for CB.
This is because the hypothesis space of the Hier-Former is larger than $\mathcal{F}_{\text{partial-nested PE}}$, which requires more training samples to reduce the estimation error.  Although with matched PE properties to the policies, the Edge-GCN cannot achieve 50\% performance even with 40,000 training samples because it does not use an attention mechanism, which is crucial for learning precoding \cite{guo2024recursive}.

\vspace{-2mm}\begin{figure}[ht]
    \centering

    \subfigure[MU-MIMO precoding]{
    \includegraphics[width=0.65\linewidth]{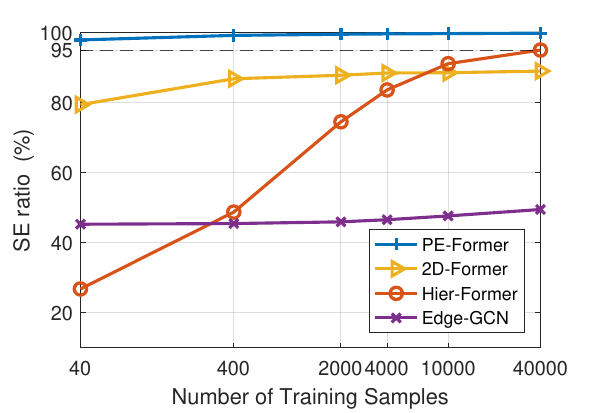}
    \label{fig:MU-MIMO performance} }\\[-1mm]
    \subfigure[CB]{
    \includegraphics[width=0.65\linewidth]{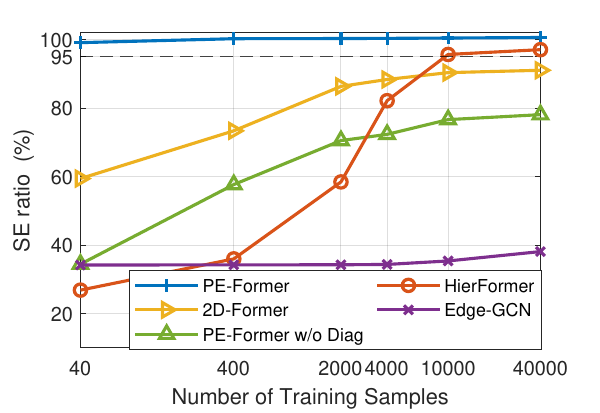}
    \label{fig:CB performance}}\\[-1mm]

    \vspace{-2mm}\caption{Learning performance versus the number of training samples.}
    \label{fig:learning performance}\vspace{-2mm}
\end{figure}

To show the impact of the diagonal output layer, we replace  $f_{\mathsf{Diag}}(\cdot)$ in the PE-Former by a sum function that does not satisfy the joint PE property, and use this DNN (denoted as ``PE-Former w/o Diag" in Fig. \ref{fig:CB performance}) to learn the CB policy. The result shows that its SE ratio is much lower than the specific PE-Former, due to the mismatch of the PE property.

\subsubsection{Size-generalizability}

This metric is evaluated on a test set containing samples with different problem sizes from the training samples. A DNN is regarded as size-generalizable if its performance tested on samples with ``unseen" sizes degrades by no more than 20\% from the performance on training samples, without re-training. The setups for generating training and test samples are provided in Table \ref{tab:generalize setup}  (see next page),
which ensures that the test set includes the samples with problem sizes not ``seen" during training. For each setup, 4,000 training samples and 1,000 test samples are generated.

\begin{table*}[htbp]
  \centering
  \caption{Setup for evaluating size-generalizability}
  \renewcommand\arraystretch{1.2}
  \vspace{-2mm}
    \begin{tabular}{c|c|c|c|c}
    \hline \hline
    Wireless policies & \multicolumn{2}{c|}{Setup} & Train set & Test set \\\hline

    \multirow{3}[0]{*}{MU-MIMO precoding}

    & \makecell{Generalizability to $N_{\mathsf{t}}$} & $K=4,N_{\mathsf{r}}=4$ & $N_{\mathsf{t}} \in \{64,128\}$ & $N_{\mathsf{t}}\in \{16,32,64,128,256,512\}$ \\\cline{2-5}
    & \makecell{Generalizability to $K$} & $N_{\mathsf{t}}=64,N_{\mathsf{r}}=4$ & $K \in \{6,7,8\}$ & $K\in \{4,\cdots,10\} $\\\cline{2-5}

    & \makecell{Generalizability to $N_{\mathsf{r}}$} & $N_{\mathsf{t}}=64,K=8$ & $N_{\mathsf{r}}\in \{3,4\}$ &  $N_{\mathsf{r}}\in \{1,\cdots,6\}$ \\

    \hline
    \multirow{3}[0]{*}{CB}

    &\makecell{Generalizability to $N_{\mathsf{t}}$} & $K=4,M=3$ & $N_{\mathsf{t}}\in \{64,128\}$ & $N_{\mathsf{t}}\in\{16,32,64,128,256,512\}$ \\\cline{2-5}

    &\makecell{Generalizability to $K$} & $N_{\mathsf{t}}=64,M=3$ & $K\in \{6,7,8\}$ & $K\in \{4,\cdots,10\}$  \\\cline{2-5}

    & \makecell{Generalizability to $M$} & $N_{\mathsf{t}}=64,K=8$ & $M\in \{3,4,5\} $& $M\in \{2,\cdots,6\}$ \\
        \hline \hline
    \end{tabular}%
  \label{tab:generalize setup}  \vspace{-3mm}
\end{table*}%

In Fig. \ref{fig: gen-mimo} (see next page), we show the size-generalizability for MU-MIMO precoding. We can see that the specific PE-Former performs the best and is well-generalizable to different values of $N_\mathsf{t}$, $K$, and $N_{\mathsf{r}}$. The 2D-Former is also size-generalizable but performs worse than the specific PE-Former, due to the approximation error.  The Hier-Former is generalizable to the values of $K$ and $N_{\mathsf{r}}$, but cannot be generalized to the values of $N_{\mathsf{t}}$.
This is because the PE property induced by the $\text{AN}^{\tt BS}$ set
is overlooked. The Edge-GCN can be generalized to different values of $N_\mathsf{t}$, $K$, and $N_{\mathsf{r}}$, but does not perform well, again due to not using an attention mechanism.

\begin{figure*}[!t]
    \centering

    \subfigure[Generalizability to $N_{\mathsf{t}}$, $K=4, N_{\mathsf{r},k}=4$]{ \includegraphics[width=0.31\linewidth]{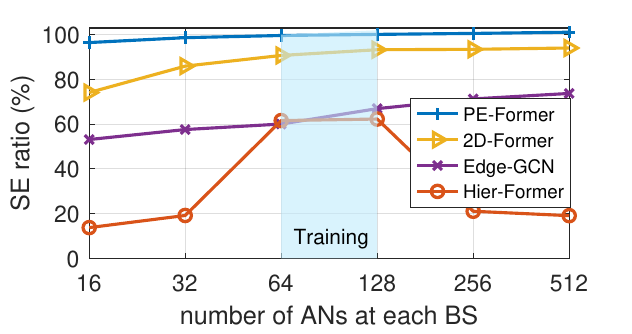}
    \label{fig:size-gen-N}}%
    \subfigure[Generalizability to $K$, $N_{\mathsf{t}}=64,N_{\mathsf{r},k}=4$]{
    \includegraphics[width=0.31\linewidth]{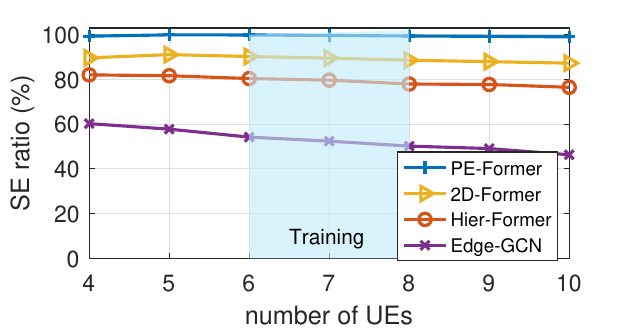}
    \label{fig:size-gen-K} }
    \subfigure[Generalizability to $N_{\mathsf{r}}$, $N_{\mathsf{t}} = 64, K=8$]{\includegraphics[width=0.31\linewidth]{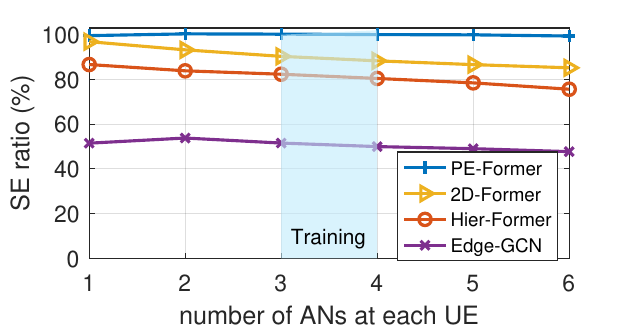} \label{fig:size-gen-S}
    }
    \vspace{-2mm}

    \caption{Size-generalizability for MU-MIMO precoding, SNR = 10 dB}
    \label{fig: gen-mimo}\vspace{-4mm}
\end{figure*}

In Fig. \ref{fig:size-gen-CB}  (see next page), we show the size-generalizability for CB. Similarly, the results show that the specific PE-Former performs the best and is well-generalizable to different values of $N_{\mathsf{t}}$, $K$, $M$. The 2D-Former is generalizable to different values of $N_{\mathsf{t}}$, $K$, $M$, but performs worse than the specific PE-Former. The Hier-Former cannot be generalized to the values of $N_{\mathsf{t}}$ and $M$ due to the overlooked PE property induced by the $\text{AN}^{\tt BS}$ set, whose size depends on the values of $N_{\mathsf{t}}$ and $M$. The Edge-GCN is size-generalizable but performs the worst due to the lack of attention mechanism.






\begin{figure*}[!t]
    \centering

    \subfigure[Generalizability to $N_{\mathsf{t}}$, $K=4, M=3$]{ \includegraphics[width=0.31\linewidth]{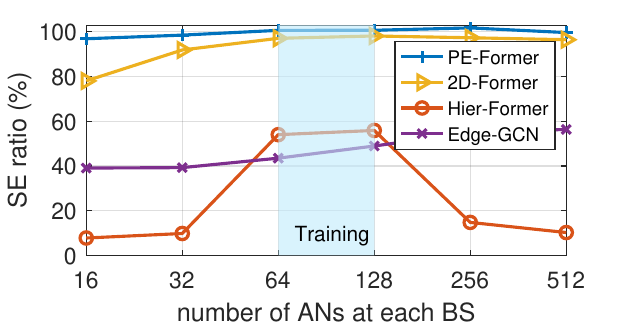}
    \label{fig:size-gen-N-CB}}%
    \subfigure[Generalizability to $K$, $N_{\mathsf{t}}=64,M=3$]{
    \includegraphics[width=0.31\linewidth]{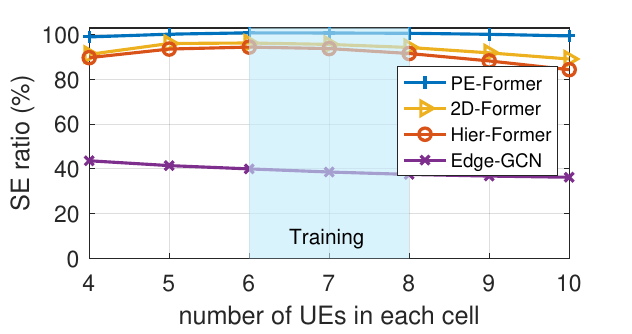}
    \label{fig:size-gen-K-CB} }
    \subfigure[Generalizability to $M$, $N_{\mathsf{t}} = 64, K=8$]{\includegraphics[width=0.31\linewidth]{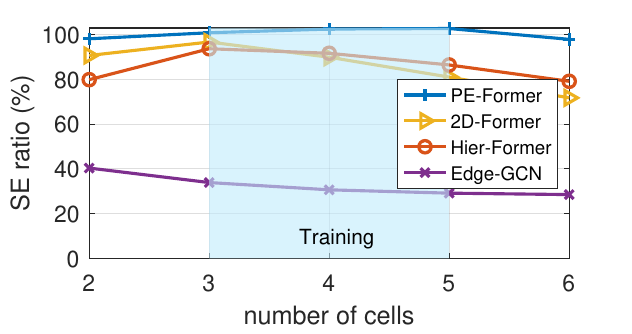} \label{fig:size-gen-M-CB}
    }
\vspace{-2mm}
    \caption{Size-generalizability for CB, SNR = 10 dB}
    \label{fig:size-gen-CB}\vspace{-2mm}
\end{figure*}

\vspace{-0.1mm}\subsubsection{Generalizability to SNRs and  Channels}
To evaluate the generalizability of the specific PE-Former to SNRs and channel statistics, we take learning MU-MIMO precoding as an example. A specific PE-Former is trained with 4,000 samples generated from the SV channel model at SNR = 10 dB, and then tested on 1,000 samples generated with different SNRs and Rayleigh channel without re-training. The SE ratios achieved in different scenarios are presented in Table \ref{tab:gen_to_channel}. The result shows that the specific PE-Former is well-generalizable to SNRs and channels.

\vspace{-3mm}
\subsection{PE-MoFormer}\vspace{-1mm}
We evaluate the performance of the PE-MoFormer and compare it with a LLM-based method. During training, all samples in the same batch are generated for the same task, and the batch size is set to 32. We no longer provide the generalization performance of the PE-MoFormer, which is almost the same as the specific PE-formers.

\subsubsection{Multi-task Learning} In Table \ref{tab:joint-training}, we provide the performance of the PE-MoFormer via joint training. In addition to five SE-maximal resource allocation tasks (i.e., narrow- and wide-band MU-MISO precoding, MU-MIMO precoding, CB, and power allocation), we also consider channel estimation in the MU-MISO
system, whose learning performance is measured by the mean square error (MSE), and the modules of ATT and FFN are trained with labels. We also provide the space complexity, i.e., the number of trainable parameters required for achieving 98\% SE ratio for each resource allocation task or the same MSE as the linear minimal mean square error (LMMSE) channel estimator. The training set consists of 100 samples for each task.
The loss function is set as the negative weighted sum of losses from all tasks. To avoid tasks with high losses dominating the training process, the weights are chosen such that the magnitudes of the weighted losses are comparable across tasks.
For comparison, we also train a specific PE-Former for each task with 100 samples and provide their performance.

\begin{table}[!t]
  \centering
  \caption{Generalizability to SNRs and channels}
  \vspace{-2mm}
    \begin{tabular}{c|c|c|c|c}
    \hline\hline
    \multirow{2}{*}{Channel model} & \multicolumn{4}{c}{SNRs} \\
\cline{2-5}          & 5 dB   & 10 dB  & 15 dB  & 20 dB \\
   \hline
    SV channel & 97.92\% & 99.61\% & 98.76\% & 96.17\% \\
    \hline
    Rayleigh channel & 99.49\% & 99.85\% & 99.42\% & 98.30\% \\
   \hline\hline
    \end{tabular}%
  \label{tab:gen_to_channel} \vspace{-5mm}
\end{table}%

\begin{table}[htbp]
  \centering
 \caption{Performance and space complexity on multiple tasks}
 \vspace{-2mm}
    \begin{tabular}{c|c|c|c|c}
    \hline\hline
    \multirow{3}{*}{Wireless tasks}      & \multicolumn{2}{c|}{PE-MoFormer} & \multicolumn{2}{c}{Specific PE-Formers} \\
    \cline{2-5}
     & \makecell{MSE/ \\ SE ratio}   & Space & \makecell{MSE/\\SE ratio}  & Space \\
    \hline
    Channel estimation*   & 0.026 &\multirow{7}{*}{9.71k} & 0.026 & 1.23k \\
    \cline{1-2}\cline{4-5}
MU-MISO precoding & 99.42\% &  & 99.51\% & 3.02k \\
\cline{1-2}\cline{4-5}
MU-MIMO precoding & 99.07\% &       & 99.31\% & 8.48k \\
\cline{1-2}\cline{4-5}
CB & 100.21\% &       & 100.48\% & 8.48k \\
\cline{1-2}\cline{4-5}
Power allocation & 99.13\% &       & 99.22\% & 1.23k \\
\cline{1-2}\cline{4-5}
\makecell{Wideband MU-MISO\\precoding} & 99.50\% &       & 99.62\% & 3.02k \\

    \hline\hline
    \multicolumn{5}{l}{\makecell[l]{ * MSE of the LMMSE  algorithm is 0.089.}}\\
    \end{tabular}%
  \label{tab:joint-training}  \vspace{-3mm}
\end{table}%


The results show that the PE-MoFormer can achieve 99\% SE ratios on all resource allocation tasks and an MSE of 0.026 on channel estimation that is lower than the LMMSE estimator. Compared to the specific PE-Formers, the PE-MoFormer exhibits a slight performance drop (within 0.3\%) due to the conflict of tasks. However, the overall space complexity of specific PE-Formers grows with the number of tasks (i.e., 1.23k + ... +3.02k = 25.48k), which is 25.48 / 9.71  = 2.62 times over the PE-MoFormer.

\subsubsection{Cross-task Learning} The PE-MoFormer is first pre-trained using 100 samples from MU-MISO precoding and CB tasks. Then, it is fine-tuned using 100 samples and tested with the samples generated for MU-MIMO and wideband MU-MISO precoding tasks. This ensures that all modules required for the two tested tasks have been pre-trained. The SE ratios achieved on the two tasks versus the number of epochs for fine-tuning are shown in Fig. \ref{fig:fine-tune}. For comparison, we also present the results of the PE-MoFormer trained for the two tested tasks from scratch with 100 samples.

The result shows that the pre-trained PE-MoFormer performs well with much fewer epochs for fine-tuning than the PE-MoFormer trained from scratch. This indicates that the pre-trained PE-MoFormer can quickly adapt to new wireless tasks ``unseen'' during pre-training.

\vspace{-0.2mm}\begin{figure}[ht]
    \centering
    \includegraphics[width=0.75\linewidth]{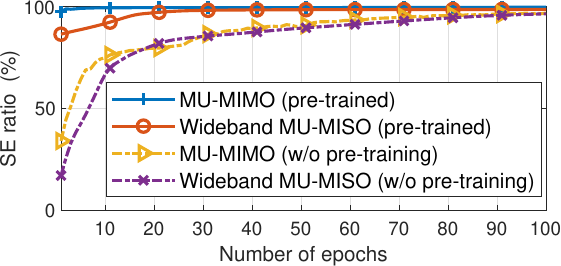}
    \vspace{-3mm}
    \caption{Learning performance versus the number of epochs}
    \label{fig:fine-tune}\vspace{-5mm}
\end{figure}

\subsubsection{Comparison with a LLM-based DNN} Finally, we compare with the DNN in \cite{DaiLL_LLMMultitask}, which adopts an LLM (e.g., GPT-2) as the backbone.
We only provide results on the model-driven MU-MISO precoding task considered therein, where the LLM actually learns the mapping from $\mathbf{H}$ to $\mathbf{P}$ and $\mathbf{\Lambda}$.
Specifically, a Transformer is used as the task-specific input encoder while an FNN and the formula in \eqref{eq: duality} are used as the task-specific output decoder in \cite{DaiLL_LLMMultitask}.
For a fair comparison, we also learn the same mapping by adding a sum-pooling after going through the selected modules for MU-MISO precoding in the PE-MoFormer in order to satisfy the PI property to ANs, and then recover $\mathbf{V}$ with \eqref{eq: duality}.

In Table \ref{tab:LLM}, we provide the SE ratios, space complexity, and inference time averaged on all test samples of the two DNNs. We can see that the space and inference complexities of the PE-MoFormer are significantly lower than the LLM-based DNN to achieve nearly the same performance.

 \vspace{-3mm}\begin{table}[htbp]
  \centering
  \caption{Comparison with a LLM-based method}
  \vspace{-2mm}
    \begin{tabular}{c|c|c|c}
    \hline\hline
    DNNs & SE ratio & Space & Time \\
    \hline
    \makecell{LLM-based \\ method }& 99.95\% & \makecell{ 6M + 118M +0.1M = 124.1M \\ (Encoder + GPT-2 + Decoder)} & 12.62 ms\\
    \hline
    \makecell{PE-MoFormer }& 99.93\% & 9.71k& 4.05 ms\\
    \hline\hline

    \end{tabular}%
  \label{tab:LLM} \vspace{-3mm}
\end{table}%

\vspace{-2mm}
\section{Conclusion}\label{sec:conclusion}
In this paper, we proposed a PE-MoFormer, a compositional DNN for learning a wide range of wireless problems with diverse PE properties. Specifically, we designed three modules that satisfy the nested 1D-PE, partial-nested 2D-PE, and joint 2D-PE properties, respectively. These modules, together with previously-designed modules, can be flexibly composed into 10 specific PE-Formers, each of which can be applied to the policies with one type of PE property. Building on the plug-and-play modules, we devised the PE-MoFormer that assembles two or three modules chosen from the seven modules by a selector, enabling a single model to learn over 25 wireless policies with 10 types of PE properties and without PE property. Simulation results showed that the PE-MoFormer performs well on multiple wireless tasks with low complexity, and exhibits strong generalizability across different problem scales, SNRs, and channel statistics.

This compositional design framework provides a foundation for extensions: new modules can be seamlessly inserted to accommodate more wireless problems with other PE properties when needed. Future work includes further reducing inference complexity and exploring the applicability of this framework to broader classes of wireless problems.

\begin{appendices} \numberwithin{equation}{section}
		\renewcommand{\thesectiondis}[2]{\Alph{section}}
\section{Proof of Proposition \ref{prop:compound_func}}\label{proof:compound_func}
\vspace{-1mm}
Without loss of generality, consider $\mathcal{F}_1$ is $\mathcal{F}_{\text{ind. 2D-PE}}$ and $\mathcal{F}_1 \subsetneq \mathcal{F}_2$. In other words, $g_{\Theta_1}(\cdot)$ satisfies the ind. 2D-PE property, and $f_{\Theta_2}(\cdot)$ satisfies one type of the PE properties in  \eqref{eq:pre_1d_pe}, \eqref{eq:pre_nested_1d_pe} and \eqref{eq:pre_joint_2d_pe} - \eqref{eq:pre_nested_joint_2d_pe}. For notational simplicity, we omit the subscripts of $f_{\Theta_2}(\cdot)$ and $g_{\Theta_1}(\cdot)$.



If $f(\cdot)$ satisfies the 1D-PE property in \eqref{eq:pre_1d_pe}, then $f(g(\mathbf{X}\mathbf{\Pi})) \overset{(a)}= f(g(\mathbf{X})\mathbf{\Pi}) \overset{(b)}{=} f(g(\mathbf{X}))\mathbf{\Pi}$, where $(a)$ holds because $g(\cdot)$ satisfies the ind. 2D-PE property in \eqref{eq:pre_2d_pe} (thus also satisfies the 1D-PE property, see Fig. \ref{fig:hypo}), and $(b)$ holds because $f(\cdot)$ satisfies the 1D-PE property. Hence, $f(g(\cdot))$ has the 1D-PE property.  Since $f(\mathbf{\Pi}_\mathsf{A}^{\mathsf{T}}\mathbf{X}\mathbf{\Pi}_{\mathsf{B}}) \neq \mathbf{\Pi}_\mathsf{A}^{\mathsf{T}}f(\mathbf{X})\mathbf{\Pi}_{\mathsf{B}}$,  $f(g(\mathbf{\Pi}_\mathsf{A}^{\mathsf{T}}\mathbf{X}\mathbf{\Pi}_\mathsf{B})) = f(\mathbf{\Pi}_\mathsf{A}^{\mathsf{T}}g(\mathbf{X})\mathbf{\Pi}_\mathsf{B}) \neq \mathbf{\Pi}_\mathsf{A}^{\mathsf{T}}f(g(\mathbf{X}))\mathbf{\Pi}_\mathsf{B}$, indicating that $f(g(\cdot))$ does not satisfy the 2D-PE property.

If $f(\cdot)$ satisfies the joint 2D-PE property in \eqref{eq:pre_joint_2d_pe}, then
$f(\mathbf{\Pi}_\mathsf{A}^{\mathsf{T}}g(\mathbf{X})\mathbf{\Pi}_\mathsf{B}) = \mathbf{\Pi}_\mathsf{A}^{\mathsf{T}}f(g(\mathbf{X}))\mathbf{\Pi}_\mathsf{B}$ holds if and only if $\mathbf{\Pi}_\mathsf{A} =\mathbf{\Pi}_\mathsf{B}$. Specifically, when  $\mathbf{\Pi}_\mathsf{A} =\mathbf{\Pi}_\mathsf{B} = \mathbf{\Pi}$, $f(g(\mathbf{\Pi}^{\mathsf{T}}\mathbf{X}\mathbf{\Pi}))\overset{(a)}{=}f(\mathbf{\Pi}^{\mathsf{T}}g(\mathbf{X})\mathbf{\Pi}) =\mathbf{\Pi}^{\mathsf{T}}f(g(\mathbf{X}))\mathbf{\Pi}$, where $(a)$ holds because $g(\cdot)$ satisfies the ind. 2D-PE property. When $\mathbf{\Pi}_\mathsf{A} \neq \mathbf{\Pi}_\mathsf{B}$, $f(g(\mathbf{\Pi}_\mathsf{A}^{\mathsf{T}}\mathbf{X}\mathbf{\Pi}_\mathsf{B})) = f(\mathbf{\Pi}_\mathsf{A}^{\mathsf{T}}g(\mathbf{X})\mathbf{\Pi}_\mathsf{B}) \neq \mathbf{\Pi}_\mathsf{A}^{\mathsf{T}}f(g(\mathbf{X}))\mathbf{\Pi}_\mathsf{B}$. Hence, $f(g(\cdot))$ satisfies the joint 2D-PE property but not the ind. 2D-PE property.

If $f(\cdot)$ satisfies the partial-nested 2D-PE property in \eqref{eq:pre_partial_nested_2d_pe}, then
$f(\mathbf{\Pi}_\mathsf{A}^{\mathsf{T}}g(\mathbf{X})\mathbf{\Pi}_\mathsf{B}) = \mathbf{\Pi}_\mathsf{A}^{\mathsf{T}}f(g(\mathbf{X}))\mathbf{\Pi}_\mathsf{B}$ holds if and only if $\mathbf{\Pi}_\mathsf{B} =\mathbf{\Omega}$. Specifically, when  $\mathbf{\Pi}_\mathsf{B} = \mathbf{\Omega}$, $f(g(\mathbf{\Pi}^{\mathsf{T}}_{\mathsf{A}}\mathbf{X}\mathbf{\Omega})) = f(\mathbf{\Pi}^{\mathsf{T}}_{\mathsf{A}}g(\mathbf{X})\mathbf{\Omega})=\mathbf{\Pi}^{\mathsf{T}}_{\mathsf{A}}f(g(\mathbf{X}))\mathbf{\Omega}$. When $\mathbf{\Pi}_\mathsf{B} \neq \mathbf{\Omega}$ , $f(g(\mathbf{\Pi}_\mathsf{A}^{\mathsf{T}}\mathbf{X}\mathbf{\Pi}_\mathsf{B})) \neq \mathbf{\Pi}_\mathsf{A}^{\mathsf{T}}f(g(\mathbf{X}))\mathbf{\Pi}_\mathsf{B}$. Hence, $f(g(\cdot))$ satisfies the partial-nested 2D-PE property but not ind. 2D-PE property.

When $f(\cdot)$ satisfies the PE properties in \eqref{eq:pre_nested_1d_pe}, \eqref{eq:pre_nested_indepedent_2D_pe}-\eqref{eq:pre_nested_joint_2d_pe},  we can also prove that $f(g(\cdot))$ has the same PE property as $f(\cdot)$ but not the ind. 2D-PE property, using similar derivations.




\section{Proof of Proposition \ref{prop:nested-pe}}\label{proof:nested-pe}
\vspace{-2mm}
We first define a mask matrix $\mathbf{M}$ as follows, \vspace{-1mm}
\begin{equation}
    \mathbf{M} \triangleq \mathbf{I}_{\mathsf{sub}} \otimes \mathbf{1}_{N_{\mathsf{s}}} =  \mathrm{diag}(\mathbf{1}_{N_{\mathsf{s}}},\cdots,\mathbf{1}_{N_{\mathsf{s}}})
\end{equation}
where $\mathbf{1}_{N_{\mathsf{s}}}$ represents a $N_{\mathsf{s}} \times N_{\mathsf{s}}$ all-one matrix.

It is not hard to prove that \vspace{-1mm}
\begin{equation}\label{PiM}
\mathbf{\Pi}\mathbf{M} \neq \mathbf{M}\mathbf{\Pi}
\end{equation}

\begin{lemma}\label{lemma:1}
For any nested permutation matrix $\mathbf{\Omega} = (\mathbf{\Pi}_{\mathsf{sub}}\otimes\mathbf{I}_{N_\mathsf{sub}})\mathrm{diag}(\mathbf{\Pi}_{1},\cdots,\mathbf{\Pi}_{N_\mathsf{sub}})$, $ \mathbf{M}\mathbf{\Omega} = \mathbf{\Omega}\mathbf{M}$ holds.
\end{lemma}
\begin{proof}
 We first prove the following equality, \vspace{-1mm}
\begin{equation} \label{proofeq:1}  \mathbf{M}(\mathbf{\Pi}_{\mathsf{sub}}\otimes\mathbf{I}_{N_\mathsf{sub}}) = (\mathbf{\Pi}_{\mathsf{sub}}\otimes\mathbf{I}_{N_\mathsf{sub}})\mathbf{M}
\end{equation}

 The left-hand side of \eqref{proofeq:1} can be re-written as $\mathbf{M}(\mathbf{\Pi}_{\mathsf{sub}}\otimes\mathbf{I}_{N_\mathsf{sub}}) = (\mathbf{I}_{N_\mathsf{sub}} \otimes \mathbf{1}_{N_{\mathsf{s}}})\cdot(\mathbf{\Pi}_{\mathsf{sub}}\otimes\mathbf{I}_{N_\mathsf{sub}}) \overset{(a)}{=} ( \mathbf{I}_{\mathsf{sub}} \cdot \mathbf{\Pi}_{\mathsf{sub}}) \otimes(\mathbf{1}_{N_{\mathsf{s}}} \cdot \mathbf{I}_{\mathsf{sub}}) = \mathbf{\Pi}_{\mathsf{sub}} \otimes \mathbf{1}_{N_{\mathsf{s}}}$, where $(a)$ is due to $(\mathbf{A}\otimes\mathbf{B})\cdot(\mathbf{C}\otimes\mathbf{D}) = (\mathbf{A}\cdot\mathbf{C})\otimes(\mathbf{B}\cdot \mathbf{D})$. Similarly, the right-hand side  can be re-written as $(\mathbf{\Pi}_{\mathsf{sub}}\otimes\mathbf{I}_{N_\mathsf{sub}})\mathbf{M} = (\mathbf{\Pi}_{\mathsf{sub}}\otimes\mathbf{I}_{N_\mathsf{sub}})\cdot(\mathbf{I}_{N_\mathsf{sub}} \otimes \mathbf{1}_{N_{\mathsf{s}}}) =(\mathbf{\Pi}_{\mathsf{sub}} \cdot \mathbf{I}_{\mathsf{sub}}) \otimes(\mathbf{I}_{\mathsf{sub}}\cdot\mathbf{1}_{N_{\mathsf{s}}}) = \mathbf{\Pi}_{\mathsf{sub}} \otimes \mathbf{1}_{N_{\mathsf{s}}}$, which equals the left-hand side.

 Since $\mathbf{1}_{N_\mathsf{s}}$ is an all-one matrix,  $\mathbf{1}_{N_\mathsf{s}} \mathbf{\Pi} = \mathbf{\Pi}\,\mathbf{1}_{N_\mathsf{s}}$ for any permutation matrix $\mathbf{\Pi}$. Hence, \vspace{-1mm}
 \begin{equation}\label{proofeq:2}
 \begin{aligned}
          &\mathbf{M} \cdot \mathrm{diag}(\mathbf{\Pi}_{1},\cdots,\mathbf{\Pi}_{N_\mathsf{sub}}) = \mathrm{diag}(\mathbf{1}_{N_\mathsf{s}}\mathbf{\Pi}_{1},\cdots,\mathbf{1}_{N_\mathsf{s}}\mathbf{\Pi}_{N_\mathsf{sub}}) \\
          & = \mathrm{diag}(\mathbf{\Pi}_{1}\mathbf{1}_{N_\mathsf{s}},\cdots,\mathbf{\Pi}_{N_\mathsf{sub}}\mathbf{1}_{N_\mathsf{s}}) = \mathrm{diag}(\mathbf{\Pi}_{1},\cdots,\mathbf{\Pi}_{N_\mathsf{sub}})\mathbf{M}
 \end{aligned}
 \end{equation}

Then, we can derive that
\begin{equation}
\begin{aligned}
    \mathbf{M\Omega} &= \mathbf{M}(\mathbf{\Pi}_{\mathsf{sub}}\otimes\mathbf{I}_{N_\mathsf{sub}}) \mathrm{diag}(\mathbf{\Pi}_{1},\cdots,\mathbf{\Pi}_{N_\mathsf{sub}}) \\
    &\overset{(a)}{=}(\mathbf{\Pi}_{\mathsf{sub}}\otimes\mathbf{I}_{N_\mathsf{sub}})\mathbf{M}\mathrm{diag}(\mathbf{\Pi}_{1},\cdots,\mathbf{\Pi}_{N_\mathsf{sub}})\\
    &\overset{(b)}{=} (\mathbf{\Pi}_{\mathsf{sub}}\otimes\mathbf{I}_{N_\mathsf{sub}})\mathrm{diag}(\mathbf{\Pi}_{1},\cdots,\mathbf{\Pi}_{N_\mathsf{sub}})\mathbf{M}=\mathbf{\Omega M} \notag
\end{aligned}
\end{equation}
where $(a)$ and $(b)$ come from \eqref{proofeq:1} and \eqref{proofeq:2}, respectively.
\end{proof}

Denote $\mathbf{C}^{(\ell)} = [\mathbf{c}_{1_1}^{(\ell)},\cdots,\mathbf{c}_{1_{N_{\mathsf{s}}}}^{(\ell)},\cdots,\mathbf{c}_{N_{\mathsf{sub}{N_{\mathsf{s}}}}}^{(\ell)}]$ and $\mathbf{D}^{(\ell)} = [\mathbf{d}_{1_1}^{(\ell)},\cdots,\mathbf{d}_{1_{N_{\mathsf{s}}}}^{(\ell)},\cdots,\mathbf{d}_{N_{\mathsf{sub}{N_{\mathsf{s}}}}}^{(\ell)}]$ and $\mathbf{D}^{(\ell)} $ as the input and output representations of all tokens in NATT. The input-output relation of NATT in \eqref{eq:2d-nested-att} can be expressed in matrix form as,
\begin{equation}\label{proofeq:def}
    \begin{aligned}
        \mathbf{C}^{(\ell)} &= \mathbf{U}_\mathsf{S}^{\mathsf{V}}\mathbf{D}^{(\ell-1)}\Big(\big(\mathbf{U}_\mathsf{S}^{\mathsf{K}}\mathbf{D}^{(\ell-1)}\big)^\mathsf{T}\mathbf{D^{(\ell-1)}} \odot \mathbf{M}\Big)\\
        & \quad+ \mathbf{U}_\mathsf{D}^{\mathsf{V}}\mathbf{D}^{(\ell-1)}\big(\mathbf{U}_\mathsf{D}^{\mathsf{K}}\mathbf{D}^{(\ell-1)}\big)^\mathsf{T}\mathbf{D^{(\ell-1)}}\\
        &\triangleq F(\mathbf{D}^{(\ell-1)}) + G(\mathbf{D}^{(\ell-1)})
    \end{aligned}
\end{equation}
where $F(\cdot)$ and $G(\cdot)$ respectively represent the local attention and global attention mechanism, where $G(\cdot)$ satisfies the 1D-PE property as proved in \cite{mehrabian2024joint}.

Next, we prove that $F(\cdot)$ is with the nested 1D-PE property, i.e., $F(\mathbf{D}^{(\ell-1)}\mathbf{\Omega}) = F(\mathbf{D}^{(\ell-1)})\mathbf{\Omega}$.

According to the definition of $F(\cdot)$ in \eqref{proofeq:def}, we have, \vspace{-1mm}
\begin{equation}\label{eq:subtodef}
\begin{aligned}
&F(\mathbf{D}^{(\ell-1)}\mathbf{\Omega})
        \!=\!\mathbf{U}_\mathsf{S}^{\mathsf{V}}\mathbf{D}^{(\ell-1)}\mathbf{\Omega}\Big(\big(\mathbf{U}_\mathsf{S}^{\mathsf{K}}\mathbf{D}^{(\ell-1)}\mathbf{\Omega}\big)^\mathsf{T}\mathbf{D^{(\ell-1)}}\mathbf{\Omega} \!\odot\!\mathbf{M}\Big)\\
        &=\mathbf{U}_\mathsf{S}^{\mathsf{V}}\mathbf{D}^{(\ell-1)}\mathbf{\Omega}\Big(\mathbf{\Omega}^{\mathsf{T}}\big(\mathbf{U}_\mathsf{S}^{\mathsf{K}}\mathbf{D}^{(\ell-1)}\big)^\mathsf{T}\mathbf{D}^{(\ell-1)}\mathbf{\Omega} \!\odot\! \mathbf{M}\Big)
        \end{aligned}
        \end{equation}

For any permutation matrices $\mathbf{\Pi}_1$ and $\mathbf{\Pi}_2$, it is not hard to prove that \vspace{-1mm}
    \begin{equation}\label{proofeq: hadamard}
        \mathbf{\Pi}_1(\mathbf{A} \odot \mathbf{B}) \mathbf{\Pi}_2 = (\mathbf{\Pi}_1\mathbf{A}\mathbf{\Pi}_2)\odot(\mathbf{\Pi}_1\mathbf{B}\mathbf{\Pi}_2)
    \end{equation}

By using \eqref{proofeq: hadamard} and Lemma \ref{lemma:1}, \eqref{eq:subtodef} can be re-written as, \vspace{-1mm}
\begin{equation}\label{eq:using_lemma}
    \begin{aligned}
        &\mathbf{U}_\mathsf{S}^{\mathsf{V}}\mathbf{D}^{(\ell-1)}\mathbf{\Omega}\Big(\mathbf{\Omega}^{\mathsf{T}}\big(\mathbf{U}_\mathsf{S}^{\mathsf{K}}\mathbf{D}^{(\ell-1)}\big)^\mathsf{T}\mathbf{D}^{(\ell-1)}\mathbf{\Omega} \!\odot\! \mathbf{M}\Big) \\
        &\overset{(a)}{=}\mathbf{U}_\mathsf{S}^{\mathsf{V}}\mathbf{D}^{(\ell-1)}\Big(\big(\mathbf{U}_\mathsf{S}^{\mathsf{K}}\mathbf{D}^{(\ell-1)}\big)^\mathsf{T}\mathbf{D}^{(\ell-1)}  \mathbf{\Omega}\odot\mathbf{\Omega}\mathbf{M}\Big)\\
        &\overset{(b)}{=}\mathbf{U}_\mathsf{S}^{\mathsf{V}}\mathbf{D}^{(\ell-1)}\Big(\big(\mathbf{U}_\mathsf{S}^{\mathsf{K}}\mathbf{D}^{(\ell-1)}\big)^\mathsf{T}\mathbf{D}^{(\ell-1)}  \mathbf{\Omega}\odot\mathbf{M}\mathbf{\Omega}\Big)\\
        &\overset{(c)}{=}\mathbf{U}_\mathsf{S}^{\mathsf{V}}\mathbf{D}^{(\ell-1)}\Big(\big(\mathbf{U}_\mathsf{S}^{\mathsf{K}}\mathbf{D}^{(\ell-1)}\big)^\mathsf{T}\mathbf{D}^{(\ell-1)}  \odot\mathbf{M}\Big)\mathbf{\Omega} \\
        &\overset{(d)}{=} F(\mathbf{D}^{(\ell-1)})\mathbf{\Omega}
    \end{aligned}
\end{equation}
where $(a)$ comes from \eqref{proofeq: hadamard} by letting $\mathbf{\Pi}_1 = \mathbf{\Omega}, \mathbf{\Pi}_2 = \mathbf{I}$, $(b)$ is due to Lemma \ref{lemma:1}, $(c)$ comes from \eqref{proofeq: hadamard} by letting $\mathbf{\Pi}_1 = \mathbf{I}, \mathbf{\Pi}_2 = \mathbf{\Omega}$, $(d)$ is due to the definition of $F(\cdot)$.

From \eqref{eq:subtodef} and \eqref{eq:using_lemma}, we have $F(\mathbf{D}^{(\ell-1)}\mathbf{\Omega}) = F(\mathbf{D}^{(\ell-1)})\mathbf{\Omega}$, i.e., $F(\cdot)$ satisfies the nested 1D-PE property.

When the nested permutation matrix $\mathbf{\Omega}$ is replaced by arbitrary permutation matrix $\mathbf{\Pi}$, $F(\mathbf{D}^{(\ell-1)}\mathbf{\Pi}) \neq F(\mathbf{D}^{(\ell-1)})\mathbf{\Pi}$ due to \eqref{PiM}, i.e., $F(\cdot)$ does not satisfy the 1D-PE property.

Since $F(\cdot)$ and $G(\cdot)$ respectively satisfy the nested 1D-PE and 1D-PE property, it can be proved that the input-output relation of the NATT, i.e., $F(\cdot)+G(\cdot)$, satisfies the nested 1D-PE property, using similar derivations in Appendix \ref{proof:compound_func}.

\section{Proof of Proposition \ref{prop:sharing_params}}\label{proof:sharing_params}
\vspace{-1mm}
With the four conditions in the proposition, it has been proved in \cite{Gformer} that $G(\cdot)$ in \eqref{proofeq:def} satisfies the ind. 2D-PE property.
From the definition of $F(\cdot)$ in \eqref{proofeq:def}, we have, \vspace{-1mm}
\begin{equation}\label{proof:sub_to_def_2}
\begin{aligned}
&F(\mathbf{\Pi}^{\mathsf{T}}\mathbf{D^{(\ell-1)}}\mathbf{\Omega})\\
        &=\mathbf{U}_\mathsf{S}^{\mathsf{V}}\mathbf{\Pi}^{\mathsf{T}}\mathbf{D}^{(\ell-1)}\mathbf{\Omega}\Big(\big(\mathbf{U}_\mathsf{S}^{\mathsf{K}}\mathbf{\Pi}^{\mathsf{T}}\mathbf{D}^{(\ell-1)}\mathbf{\Omega}\big)^\mathsf{T}\mathbf{\Pi}^{\mathsf{T}}\mathbf{D^{(\ell-1)}}\mathbf{\Omega} \odot\mathbf{M}\Big)
\end{aligned}
\end{equation}
Further using the given conditions $\mathbf{\Pi}\mathbf{U}^{\mathsf{K}}_{\mathsf{S}} = \mathbf{U}^{\mathsf{K}}_{\mathsf{S}}\mathbf{\Pi}, \mathbf{\Pi}\mathbf{U}^{\mathsf{V}}_{\mathsf{S}} = \mathbf{U}^{\mathsf{V}}_{\mathsf{S}}\mathbf{\Pi}$, \eqref{proof:sub_to_def_2} can be re-written as, \vspace{-1mm}
\begin{equation}\label{eq:derive_partial_nested}
    \begin{aligned}
&\mathbf{\Pi}^{\mathsf{T}}\mathbf{U}_\mathsf{S}^{\mathsf{V}}\mathbf{D}^{(\ell-1)}\mathbf{\Omega}\Big(\big(\mathbf{\Pi}^{\mathsf{T}}\mathbf{U}_\mathsf{S}^{\mathsf{K}}\mathbf{D}^{(\ell-1)}\mathbf{\Omega}\big)^\mathsf{T}\mathbf{\Pi}^{\mathsf{T}}\mathbf{D}^{(\ell-1)}\mathbf{\Omega} \odot \mathbf{M}\Big)\\
        &\overset{(a)}{=}\mathbf{\Pi}^{\mathsf{T}}\mathbf{U}_\mathsf{S}^{\mathsf{V}}\mathbf{D}^{(\ell-1)}\mathbf{\Omega}\Big(\mathbf{\Omega}^\mathsf{T}\big(\mathbf{U}_\mathsf{S}^{\mathsf{K}}\mathbf{D}^{(\ell-1)}\big)^\mathsf{T}\mathbf{D}^{(\ell-1)}  \mathbf{\Omega}\odot\mathbf{M}\Big)\\
        &\overset{(b)}{=} \mathbf{\Pi}^{\mathsf{T}}F(\mathbf{D}^{(\ell-1)})\mathbf{\Omega}
    \end{aligned}
\end{equation}
where $(a)$ is due to $\mathbf{\Pi}\cdot \mathbf{\Pi}^{\mathsf{T}}=\mathbf{I}$, and $(b)$ is due to \eqref{eq:using_lemma}.

With \eqref{proof:sub_to_def_2} and \eqref{eq:derive_partial_nested}, we have $F(\mathbf{\Pi}^{\mathsf{T}}\mathbf{D}^{(\ell-1)}\mathbf{\Omega}) = \mathbf{\Pi}^{\mathsf{T}}F(\mathbf{D}^{(\ell-1)})\mathbf{\Omega}$, i.e., $F(\cdot)$ satisfies the parial-nested 2D-PE property in \eqref{eq:pre_partial_nested_2d_pe}. Since $G(\cdot)$ satisfies the ind. 2D-PE property, it can be proved that the input-output relation of NATT (P.S.), i.e., $F(\cdot)+ G(\cdot)$, satisfies the partial-nested 2D-PE property, using similar derivations in Appendix \ref{proof:compound_func}.

\end{appendices}

\bibliographystyle{IEEEtran}
\bibliography{IEEEexample}
\end{document}